\documentclass[11pt]{article}

\usepackage{palatino}
\fontsize{10.5}{12}\selectfont

\usepackage{parskip}

\usepackage{float}
\usepackage{titlesec}
\titleformat{\subsubsection}[runin]
  {\normalfont\large\bfseries}{\thesubsubsection}{1em}{}

\usepackage[font={small,it}]{caption}

\usepackage[pdftex]{graphicx}
\usepackage{amssymb,amsthm,amsmath}

\usepackage{graphicx}
\usepackage{mathtools}
\usepackage{lscape}
\usepackage{subcaption}
\usepackage{suffix}
\usepackage{color}
\usepackage{setspace}
\usepackage{mathrsfs}
\usepackage[inline]{enumitem}
\usepackage{framed}
\usepackage{multirow}
\usepackage{wrapfig}
\usepackage[rflt]{floatflt}

\usepackage{qtree}

\RequirePackage[hyphens]{url}
\usepackage{xcolor}
\usepackage{hyperref}
\definecolor{darkblue}{rgb}{0.0,0.0,0.3}
\hypersetup{colorlinks,breaklinks,linkcolor=darkblue,urlcolor=darkblue,
            anchorcolor=darkblue,citecolor=blue}
\setlist*[enumerate]{label=(\roman*)}

\def\boxit#1{\vbox{\hrule\hbox{\vrule\kern6pt
          \vbox{\kern6pt#1\kern6pt}\kern6pt\vrule}\hrule}}
					
\usepackage{tikz}

\usepackage[authoryear]{natbib}
\usepackage{url}
\usepackage{bigstrut}

\usepackage{nicefrac}

\numberwithin{equation}{section}

\renewcommand{\det}[1]{\left| #1 \right|}

\newcommand{\x}{{\bf x}}
\newcommand{\y}{{\bf y}}

\newcommand{\E}{{\bf E}}

\newcommand{\bbeta}{\boldsymbol{\beta}}

\newcommand{\ben}{\begin{enumerate}}
\newcommand{\een}{\end{enumerate}}
\newcommand{\beq}{\begin{equation}}
\newcommand{\eeq}{\end{equation}}
\newcommand{\bde}{\begin{description}}
\newcommand{\ede}{\end{description}}

\newcommand{\abs}[1]{\lvert#1\rvert}

\newcommand{\vectornorm}[1]{\left|\left|#1\right|\right|}

\newcommand{\NormRV}{\mathcal{N}}

\newcommand{\bone}{{\bf 1}}

\newcommand{\iid}{\stackrel{\mathrm{iid}}{\sim}}

\newtheoremstyle{slplain}
  {1\baselineskip\@plus.2\baselineskip\@minus.2\baselineskip}
  {.5\baselineskip\@plus.2\baselineskip\@minus.2\baselineskip}
  {\slshape}
  {}
  {\bfseries}
  {.}
  { }
  {}

\theoremstyle{slplain}

\newtheorem{theorem}{Theorem}

\newtheorem{lemma}[theorem]{Lemma}

\newtheorem{proposition}[theorem]{Proposition}
\newtheorem{remark}[theorem]{Remark}

\usepackage{booktabs,array}

\newcount\rowc

\graphicspath{{../../art/}{../art/}{./art/}{./}}
\graphicspath{{../../figs/}{../figs/}{./figs/}{./}}
\graphicspath{{../../aim1/}{../aim1/}{./aim1/}{./}}
\graphicspath{{../../aim2/}{../aim2/}{./aim2/}{./}}
\graphicspath{{../../aim3/}{../aim3/}{./aim3/}{./}}

\allowdisplaybreaks

\makeatletter
\def\mathcolor#1#{\@mathcolor{#1}}
\def\@mathcolor#1#2#3{%
  \protect\leavevmode
  \begingroup
    \color#1{#2}#3%
  \endgroup
}
\makeatother

\usepackage{xcolor}
\newenvironment{ans}{\color{blue} }{\color{black} }
\def \bans{\begin{ans}}
\def\eans{\end{ans}}
\newenvironment{grey}{\color{gray} }{\color{gray} }
\def\bgrey{\begin{grey}}
\def\egrey{\end{grey}}

\title{On Posterior consistency of Bayesian Changepoint models}

\author{Nilabja Guha \\ University of Massachusetts Lowell \\
Jyotishka Datta \\ Virginia Polytechnic Institute and State University}

\date{}

\begin{document}
\maketitle

\begin{abstract}

While there have been a lot of recent developments in the context of Bayesian model selection and variable selection for high dimensional linear models, there is not much work in the presence of change point in literature, unlike the frequentist counterpart. We consider a hierarchical Bayesian linear model where the active set of covariates that affects the observations through a mean model can vary between different time segments. Such structure may arise in social sciences/ economic sciences, such as sudden change of house price based on external economic factor, crime rate changes based on social and built-environment factors, and others. Using an appropriate adaptive prior, we outline the development of a hierarchical Bayesian methodology that can select the true change point as well as the true covariates, with high probability. We provide the first detailed theoretical analysis for posterior consistency with or without covariates, under suitable conditions. Gibbs sampling techniques provide an efficient computational strategy. We also consider small sample simulation study as well as application to crime forecasting applications.  

\end{abstract}

\section{Introduction}

In many applications such as economics, social science, the observed variable depends on covariates through mean structure, where the mean structure changes with time, based on changes in some latent unobserved factor such as an economic phenomenon/public policy change \citep{datta2019bayesian}. The dependence on the covariate may be local (only in some time segments) or global. Selecting the change point and covariates consistently remains an important problem which provides an insight to underlying sociological and economical factors. 

There is a huge and influential literature, both Bayesian and frequentist, on changepoint detection with application in diverse areas that dates back several decades. Initial attempts for changepoint detection using cumulative sums date back to \citet{page1955test,page1957problems}. Shortly after, changepoint for the location parameter, primarily within the Gaussian observation model, was studied by several authors including \citet{chernoff1964estimating,gardner1969detecting,srivastava1981tests, sen1973multivariate, smith1975bayesian, sen1980asymptotic}. The problem of multiple changepoints were addressed by several people, notably \citet{talwar1983detecting, stephens1994bayesian, chib1998estimation}. \citet{vostrikova1981detecting} introduced the popular binary segmentation method that recursively partitions the observation window to estimate the number of changepoints. The early survey by \citet{zacks1983survey} provides a detailed account of these early innovations. Early Bayesian methodological contributions include \citet{carlin1992hierarchical}, who provide a hierarchical Bayes framework for changepoint model with applications to changing regressions and changing Markov structures and \citet{raftery1994change}, who propose a Markov Transition Distribution model and provide Bayes factors for testing whether a change-point has occurred in a given segment. \citet{csorgo1997limit} provides a detailed review of changepoint methods and some new contributions based on likelihood based tests. There has been a resurgence of changepoint literature in the last decade with a renewed focus on both running time \citep[e.g.][]{killick2012optimal} and suitable handling of multi-resolution and multidimensional nature of modern experiments and data collection routines \citep[e.g.][]{frick2014multiscale, fryzlewicz2014wild}. 

It is worthwhile to note that the majority of Bayesian changepoint estimation literature considers an offline, retrospective approach while online changepoint detection methods are somewhat more prevalent in the frequentist regime, as pointed out by \citet{adams2007bayesian}, who propose an online Bayesian method based on \textit{recursive run length estimation}. 

With high-throughput data becoming routine in modern scientific studies, the problem of variable selection in a high-dimensional changing linear model becomes important where a key inferential goal is to identify the potentially different set of `active' variables within each segment. Recent papers that address this problem include frequentist approach such as \citet{lee2016lasso}, and Bayesian treatment in \citet{datta2019bayesian}. In \citet{lee2016lasso}, a lasso penalization approach is used for changing high-dimensional linear regression while selecting relevant regressors under sparsity assumption. While in \citet{datta2019bayesian}, this is handled by using the shrinking and diffusing prior \citep{narisetty2014bayesian} in each segment for variable selection. As \citet{datta2019bayesian} points out, the decomposability of the likelihood for changing linear regression model also makes it easy to incorporate other Bayesian variable selection priors: for example, one could use a spike-and-slab prior \citep{mitchell88} or a variety of shrinkage priors that have become quite popular for sparse variable selection and estimation \citep{polson2010shrink, bhadra2019lasso}. 

\subsection{Recent theoretical results} As the main focus of this article is theoretical guarantees, it is worthwhile to briefly mention a few recent theoretical results that are relevant to our present discourse. Roughly speaking, the recent theoretical advances can be broadly classified into two major thrusts: (a) optimality properties for the estimator for the underlying piecewise mean and (b) optimality properties for the estimator for changepoint locations. 

For the first kind, \citet{gao2017minimax} established sharp nonasymptotic risk bounds for least squares estimators when the underlying mean has a piecewise constant structure, and observed a phase change phenomenon when the number of changepoints goes beyond $2$. Let $\Theta_k$ denotes the model with all piecewise constant $\theta$ with maximum $k-1$ changepoints and $\hat{\theta}(\Theta_k)$ is the least squares estimator (LSE) under this model. Now consider a possibly misspecified LSE $\hat{\theta}(\Theta_k)$, when the true $\theta^0 \in \Theta_{k_0}$. \citet{gao2017minimax} provided sharp risk bounds that are minimax when $k = k_0$, i.e. $\inf_{\hat{\theta}} \sup_{\theta_0 \in \Theta_{k_0}} \E \vectornorm{\hat{\theta}-\theta_0} \asymp \sigma^2 \{ 1 + \log \log n \bone(k = 2) + k \log(e n/k) \bone(k > 2) \}$. \cite{martin2017asymptotically} developed an efficient empirical Bayes strategy for the piecewise constant sequence model and showed that the resulting posterior distribution attains a similar optimal rate as in \citet{gao2017minimax}. Theorem 1 of \citet{martin2017asymptotically} states that under a data-driven prior on the elements of $\theta$ and a truncated geometric prior distribution on the number of changepoints, the empirical Bayes posterior distribution $\Pi^n$ of $\theta \in \mathbb{R}^n$ will satisfy: $\sup_{\theta_0} \E_{\theta_0} \Pi^n( \{\theta \in \mathbb{R}^n : \vectornorm{\theta - \theta_0}^2 > M_n \epsilon_n(\theta_0) \}) \to 0, \quad n \to \infty$, where $M_n$ is any sequence with $M_n \to \infty$, and $\epsilon_n(\theta_0)$ is the target rate similar to the one in \citet{gao2017minimax}.
\citet{martin2017asymptotically} point out that the concentration rate for the empirical Bayes posterior distribution will have one phase transition from $k = 1$ to $k \ge 2$, unlike two phase transitions noted by \citet{gao2017minimax}, and conjecture that this phenomenon might be a characteristic of all Bayesian approaches for piecewise constant changepoint detection. \cite{liu2019minimax} extends \citet{gao2017minimax} to the case of multidimensional change-point detection, where the location $\theta \in \mathbb{R}^{p \times n}$ can change in at most $s$ out of $p$ coordinates at some time-point $t_0 \in \{1, \ldots, n\}$.  

For estimating the location of change points involving exponential families, \citet{frick2014multiscale} proposed the multiscale estimator SMUCE that attains the minimax rate $O(n^{-1})$ up to a logarithmic factor. \citet{frick2014multiscale} also constructed asymptotically honest confidence sets for the number and location of change points, and provided sufficient conditions for the SMUCE method to detect change points with probability approaching $1$ in the presence of `vanishing signals' for $n \to \infty$.

\subsection{Our contributions and outline}

Despite these remarkable advances, there is essentially no theoretical guarantees for changepoints in high-dimensional linear models concerning consistency in model selection. Here we provide the following theoretical substantiations:
\ben
\item We show that under the default Bayesian hierarchy, it is possible to recover both the true change point locations and the true non-zero covariates with high probability under mild conditions on the covariates and the maximum model size. 
\item Specifically, we prove formal posterior consistency results for model selection and change point recovery via Bayes factor for both piecewise constant model as well as changing high-dimensional linear regression. We also prove that the minimax rate of $O(n^{-1})$ \citep{frick2014multiscale} is attained by the Bayes estimators for change point recovery. 
\item Finally, we show that the empirical Bayes estimator attains the same optimal rate of convergence as the full Bayes solution, under the assumption of same, but unknown, error variance $\sigma^2$ across different segments. 
\een
To our knowledge, this is the first theoretical substantiation of the superior performance of Bayesian methods in this specific methodological context.

%

\section{Mathematical/Asymptotic Framework}\label{sec:framework}

Consider the canonical high-dimensional regression set-up with an $n$-dimensional response $y$ and an $n \times p$ design matrix $X$, with $p \gg n$, where $y_i \sim \NormRV(\x_i'\bbeta_j,\sigma^2)$ for covariate vector $\x_i$  at time point $t_i$, $i=1,\dots,n$. Let $\bbeta_{1}, \bbeta_{2}\dots, \bbeta_{l}$;  $\bbeta_{i}\neq \bbeta_{i+1}$, $i=1,\cdots,l-1$ are the values of the coefficient vector, where $\bbeta_{k+1}$ is the value of the coefficient vector between $nt_k < i\leq\lfloor nt_{k+1}\rfloor, k=1,\dots,l-1$, and $t_1,\dots,t_{l-1}$ are locations of change points and $t_0=0$, and let $\beta_{j,m}$ be the $m$ th component of $\bbeta_j$. 
 
\subsection{Changing linear model with variable selection}

\textbf{Spike-and-slab priors:} For the changing linear regression, we want to incorporate covariates in the model with selection of relevant predictors for each time segment between two changepoints. Our aim is to simultaneously select the true non-zero covariates as well as infer the correct number and positions of the changepoints. Our framework allows for using different priors that enable variable selection. The natural Bayesian solution is to put a spike-and-slab prior on $\beta_j$ that will ensure selection of covariates \citep{mitchell88}. 

Let $I_i^y$ be the indicator function associated with $y_i$ where $I_i^y=1$ denotes a change point at the $i^{th}$ epoch/location. Given $I_i^y$'s $i = 1, \ldots, n$, let $l=1+\sum_{i=1}^n I_i^y$ be the number of partition based on change points and $P_j$ denote the $j^{th}$ partition. The hierarchical model can be written as: 
\begin{align}
&y_i \mid \bbeta_j, \sigma^2   \sim \NormRV(\x_i' \bbeta_{j},\sigma^2); \; i \in P_j; \quad \pi(\sigma^2) \propto \frac{1}{\sigma^2}; \nonumber\\
&\beta_{j,m} \mid I_m^\beta, \pi, \tau_j^2  \sim (1-I_m^\beta) \delta_{\{0\}} + I_m^\beta \NormRV(0, \tau_j^2),  \text{ for } i \in P_j, \; j = 1,\ldots,l; \nonumber\\
& I_m^\beta  \sim \text{Bernoulli}(\tilde{p}_m), \text{ and }  I_i^y \sim \text{Bernoulli}(p_n) \label{model2}
\end{align}
where  $I_i^y$'s are independent Bernoulli indicators of whether the $i^{th}$ observation $y_i$ is associated with a `change-point', and $I_m^\beta$'s are independent indicator Bernoulli random variables, indicating whether the $m^{th}$ covariate is included in the true model, or analogously if the $m^{th}$ parameter in the $j^{th}$ partition, $\beta_{j,m}$ is non-zero. 


Here the variable selection can be done via the posterior inclusion probability (PIP) for each $\beta_j$ within each time segment. The inclusion probability $P(\beta_j \ne 0 \mid \y)$ is used to select the relevant predictors. Spike-and-slab priors have proven optimality properties for high-dimensional linear models as shown by \citep{castillo2015bayesian}, although they come with a higher computational burden due to the need for exploring a high-dimensional parameter space. Alternatively, the global-local shrinkage priors \citep{polson2010large, bhadra2019lasso} that have been proven optimal for variable selection \citep{datta2013asymptotic,ghosh2016testing} can be used. 

The simpler case of a Gaussian sequence model with no covariates is presented first to build the intuition, which will later be generalized for covariates.  This leads to the following hierarchical model:
\begin{align}
&y_i  \mid \theta_i,\sigma^2\sim  \NormRV(\theta_i,\sigma^2); \quad \pi(\sigma^2) \propto \frac{1}{\sigma^2}; \nonumber\\
&\theta_i \mid \mu_j,\tau^2  \sim \NormRV(\mu_j, \tau^2),  \text{ for } i \in P_j, j=1,\dots,l; \nonumber\\
&\mu_j \sim   \NormRV(0,V), \; V \gg 1 \text{ and } \pi(\tau^2 ) \propto \frac{1}{\tau^2}; \quad I_i^y \sim \text{Bernoulli}(p_n); \label{model1}
\end{align}
where $I_i^y$ is the indicator whether $y_i$ is a change point, $\mu_j$'s are independent given $I_i^y$'s, $\theta_i$'s are independent  given $\mu_j$'s and $\tau^2$, and $y_i$'s are independent given $\theta_i$ and $\sigma^2$. It is assumed that  $p_n \sim n^{-1}c_n$. For $n$ observations the expected number of change points $E[\sum_{i=1}^n I_i]=c_n$ is therefore of the order of $c_n$. Here, $c_n$ controls the  expected numbers of  change points apriori, and letting $c_n$ increase to infinity provides flexibility to model multiple and potentially large number of change points for large $n$.

\section{Theoretical Properties}\label{sec:theory}

In this section, we establish the change-point detection consistency, in presence of covariate, and establish the rates of convergence. Let $l^*$ be the true number of change-point and $l^*+1$ be the corresponding number of the partitions. Let $m_\beta^*$ denote the true model for the covariate combination. Let $M^*=M^*_{l_{t^*}^*,{m_\beta^*}}=M^*_{l^*,{m_\beta^*}}$ be denote the model with both true covariate combination and true change-point location, where for convenience we drop the suffix $t^*$ in $l^*_{t^*}$. Similarly, $M^*_{l^*,{\tilde{m}_\beta^*}}$  be a model with true change-point locations and $\tilde{m}_\beta^*\supset {m}_\beta^*$. For a generic model with change-point locations $\underline{t}_l=(t_1,\cdots,t_{l})$ and the covariate combination $m_\beta$, we write as $M_{l,{m}_\beta}^{\underline{t}_l}$ or $M_{l,{m}_\beta}$ dropping the index $\underline{t}_l$ for notational convenience.

For a model $M_{l,{m}_\beta}$ and for partition $P_j$ corresponding to change points $t_{j-1}$, and $t_j$,  let $X_j$ be the corresponding design matrix. Let $q_n=o(n)$ be the maximum number of covariates permissible in a model. Then, we can state the following relatively straightforward result about the marginal distribution.

\begin{proposition}
Under model given in \eqref{model2} we have for $\tau_j^2=O(1)$, and known $\sigma^2$
\begin{equation}
\log L({\bf Y}_n \mid M_{l,m_\beta})=-\frac{1}{2}\sum_j\log(\det{X_j'X_j})-\frac{1}{2\sigma^2}\sum_j\sum_{i\in P_j}(y_i-x_i\hat{\beta}^j)^2+c_{l,n}+O(1)
\label{margcov1}
\end{equation}
where $\hat{\beta}^j$ is the least squares estimator of the coefficient vector based only on the observations lying in $P_j$, and $c_{l,n}=-\frac{n}{2}\log \sigma^2+c_l$, $c_l=O(kl')$ where $k=\#\{m_\beta\}$ denotes the size of the model $m_{\beta}$.
\label{prop1}
\end{proposition}

\begin{remark}
For $\tau_j^2 \gg 1$ in \eqref{model2}, we have
\[\log L({\bf Y}_n|M_{l,m_\beta})=-\frac{1}{2}\sum_j\log( \det{X_j'X_j})-\frac{1}{2\sigma^2}\sum_j\sum_{i\in P_j}(y_i-x_i\hat{\beta}^j)^2+c'_{l,n},\]
where $c'_{l,n}=-\frac{n}{2}\log \sigma^2+c'_l$, $c'_l=O(kl')$. Using the above mentioned result, it will be sufficient to prove our result  for the flat normal  prior on the coefficients, which will help us  streamline the derivation of the  following results.
\end{remark}

Since many important applications as well as recent developments such as \citet{frick2014multiscale} use the equal and known variance set-up: we establish our result first for this set-up. Then we extend our results to the case of unknown variance.
%

For a model $M_{l^*,m_\beta}$ with the same number of change point as the true model let $P_j$ be the $j^{th}$ partition of $i=1,\cdots, n$. Let $P_{j^*}$ be the $j^{th}$ partition for $M^*$.  Suppose,  $X_{j\cap j^*}$ be the design matrix corresponding to $P_j\cap P_j^*$ and covariate combination $m_\beta\cup m_{\beta^*}$, and   $X_{j- j^*}$ corresponds to $P_j-P_j^*$, $X_{j\cap k^*}$ corresponds to $P_j \cap P_{k}^*$,   and  $\beta^*_i$ be the true coefficient vector for $y_i$ and  $\beta^*_i={\beta^*}^j$ for $P_j \cap P_j^*$. Note that the design matrices are constructed when $P_j\cap P_j^*, P_j-P_j^*, P_j\cap P_k^*$ are non empty, respectively.  Let  $\hat{\beta}^j$ be the least square estimator based on $P_j$ based on model covariate choice given by $m_\beta$, and the entries corresponding to $m_\beta \cup m_{\beta^*}-m_\beta$ is zero. Similarly, ${\tilde{\beta}}^*_i$ and ${\tilde{\beta}}^{*j}$ is defined by filling the entries not  in $m_{\beta^*}$ by zero.

Then, it follows after some calculations, 

\begin{equation}
\begin{split}
P({\bf Y}_n \mid M_{l^*,m_\beta}) & =-\frac{1}{2}\sum_j\log(\abs{X_j'X_j})-\frac{1}{2\sigma^2}[\sum_{i}(y_i-\theta^*_i)^2+\sum_{j:P_j\cap P_j^*\neq \phi }({\tilde{\beta}}^{*j}-E[\hat{\beta}^j])'X_{j\cap j^*}'X_{j\cap j^*}({\tilde{\beta}}^{*j}-E[\hat{\beta}^j]) \\ 
& +\sum_{j}\sum_{k:k\neq j; P_j\cap P^*_k\neq \phi}({\tilde{\beta}}^{*k}-E[\hat{\beta}^j])'X'_{j\cap k^*}X_{j\cap k^*}(\tilde{\beta}^{*k}-E[\hat{\beta}^j]) \\
&+\sum_j(\hat{\beta}^j-\E[\hat{\beta}^j])'X_j'X_j(\hat{\beta}^j-E[\hat{\beta}^j])] -\frac{1}{\sigma^2}[E_1+E_2+E_3]+c_{l,n}+O(1) 
\end{split}
\label{decomp0}
\end{equation}

where, $E_1,E_2,E_3$ are the cross product terms, and $\theta_i^*$ is the true value of $E[y_i]$. Hence, accounting for the bias terms and bounding the cross product terms we can show the Bayes factor consistency if the proposed partition is not a refinement of the true partition for some covariate combination. If $l>l^*$ then the expression similar to in equation \ref{decomp0} can be derived by comparing $M_{l,m_\beta}$ with model with  partition/change points corresponding to a refinement of the true partition, with $l$ change points. A refinement of the true partition is a partition corresponding to change points $\tilde{t}_1,\dots,\tilde{t}_{l}$ , with $l>l^*$ and $t^*_i\in \{\tilde{t}_1,\dots,\tilde{t}_{l}\}$, for $i=1,\dots, l^*$.

If the proposed partition is a refinement of the true partition, and the covariate combination contains the true set of covariates, then we will have terms asymptotically similar to Schwarz's BIC \citep{schwarz1978estimating} which will guarantee Bayes factor consistency. 

Next, we make the following assumptions on covariates and model size:

\begin{itemize}

\item[(A1)] The covariates $X_j$'s are uniformly bounded.

\item[(A2)] Coefficients $\beta_j$'s are uniformly bounded below and above.

\item[(A3)]  {$(q_n\log n)^2\prec n\epsilon_n$} and {$q^2_n\preceq {n}^c, c<0.5$}  ; for a sequence $\epsilon_n\rightarrow 0$ such that $n\epsilon_n\rightarrow \infty$. For two positive sequences $\{a_n\}_{n\geq 1},\{b_n\}_{n\geq 1}$, $a_n \prec b_n$, if $a_n/b_n \rightarrow 0$ as $n\uparrow\infty$. 

\item[(A4)] {$\tau_j^2=V_n=O(1)$}.
\item[(A5)] {$0<t_1^*<\cdots<t_{l^*}^*<1$} are the true locations of change points (i.e. $  nt_j^*\leq \tau^*_j \leq nt_j^*+1$).

\item[(A6)]  Let $\tilde{P}_1,\tilde{P}_2,\cdots, \tilde{P}_k$ be $k$ disjoint subset of $\{1,\dots,i\dots,n\}$ of size $n_1$, $n_2,\cdots, n_k$. Let $X_1, X_2 \cdots, X_k$ be design matrices corresponding to observations indexed by $\tilde{P}_1\cdots \tilde{P}_k$ for some model based on $q_n=o(n)$ many coefficients.

For, $n_1,\dots, n_k\rightarrow \infty, q_n/n_1,\dots, q_n/n_k\rightarrow 0$,  assume { $an_i<\lambda_{A}<bn_i, a \frac{n_1}{n_1+n_2+\cdots+n_k}<\lambda_C<b\frac{n_1}{n_1+n_2+\cdots+n_k}$} where $\lambda_X$ is singular value of $X$  and,  $A=X_i'X_i$ and $C=(X_1'X_1+X_2'X_2+\cdots+X_k'X_k)^{-1}X_i'X_i$, and $a,b>0$.
\end{itemize}

We can now state the following result for the number of change points:
\begin{theorem}
Let $0<t_1^*<\cdots<t_{l^*}^*<1$ be the locations of true change points, $M^*$ be the true model corresponding to true covariates and true change points, and $0<t_1<\cdots<t_l<1$ be the locations  of change points for the alternative with  $l \neq l^*$. Let $M_l$ be the corresponding model with change points at $t_1<\cdots<t_l$ and  some covariate combinations $m_\beta$ with at most $q_n$ many covariates. Then under $A1-A6$, for \eqref{model2}  \[ BF(M_{l,m_\beta},M^*)\rightarrow 0 \text{ in probability as } n \rightarrow \infty .\] 
\label{thm1}
\end{theorem}

Let $0<t_1^*<\cdots<t_{l^*}^*<1$ be the locations of true change points and $M^*$ be the true model corresponding to true covariates and true change points. Let $0<t_1<\cdots<t_{l^*}<1$ be the location  of change points for an alternative model with some covariate combination and assume $\max_i |t_i^*-t_i|>\epsilon_n$ . Let $M^{\epsilon_n}_{l^*}$ be the corresponding model with change points at $t_1,\cdots,t_{l^*}$ and  some covariate combinations $m_\beta$ with at most $q_n$ many covariates. Then, we show that even for $\epsilon_n\sim n^{-1}$ up to some log factors, if $q_n$ increases in logarithmic rate, we have consistency., where  for  two positive sequences $\{a_n\}_{n\geq 1},\{b_n\}_{n\geq 1}$, $a_n\sim b_n$ if $a_n/b_n$ is bounded away from zero and infinity. In particular:

\begin{theorem}
Under $A1-A6$, for \eqref{model2}   \[BF(M^{\epsilon_n}_{l^*,m_\beta}, M^*) \rightarrow 0 \text{ in probability as } n \to \infty .\] 
\label{thm2}
\end{theorem}

Note that {$q_n\log n\prec \sqrt{n\epsilon_n}$}. If  $q_n=O(\log n)$,  {$\epsilon_n\sim n^{-1}$} (up to logarithmic factors), satisfies this condition.
Under stronger conditions,  the Bayes factor consistency holds uniformly over covariate choice. The results hold for misspecification of variance parameter.
%

\begin{remark}
The condition of bounded covariates given in $(A1)$ can be relaxed. For example, if we use sub-exponential type tail bound conditions on the distributions of the covariates, that is, if $F_m(\cdot)$ is the CDF for $m^{th}$ covariate and $-\log (1-F_m(t))\succeq t^\alpha$ for some $\alpha>0$, then  Theorem \ref{thm2} holds under slightly modified (A3) (up to log factors). For two positive sequences $\{a_n\}_{n\geq 1},\{b_n\}_{n\geq 1}$, $a_n\preceq b_n$ if $a_n\leq Kb_n$ for some constant $K>0$.
\end{remark}

Instead of using known variance $\sigma^2$, if  we  use $\hat{\sigma}^2=n^{-1}\sum_{j}\|Y_j-\hat{Y}_j\|^2$, it can be shown that the conclusion of Theorems \ref{thm1} and \ref{thm2} hold.  Here $Y_j$ is the observation vector for $P_j$ and the $\hat{Y}_j$ is the least square fit for the proposed model $m_\beta$ based on the observations and covariates in  $P_j$. This result is addressed in the following Theorems.

%

\begin{theorem}
Let $0<t_1^*<\cdots<t_{l^*}^*<1$ be the locations of true change points, $M^*$ be the true model corresponding to true covariates and true change points, and $0<t_1<\cdots<t_l<1$ be the locations  of change points for the alternative with  $l \neq l^*$. Let $M_l$ be the corresponding model with change points at $t_1<\cdots<t_l$ and  some covariate combinations with at most $q_n$ many covariates. Then under $A1-A6$, for  \eqref{model2}, for the empirical estimator of $\sigma^2$,                      
\[BF(M_{l,m_\beta},M^*) \rightarrow 0 \text{ in probability as } n \rightarrow \infty .\] 
\label{thm1_emp}
\end{theorem}
For the rate calculation, we have the following result. 
\begin{theorem}
Under $A1-A7$, for \eqref{model2} , for the empirical estimator of $\sigma^2$,  \[BF(M^{\epsilon_n}_{l^*}, M^*) \rightarrow 0 \text{ in probability as } n \rightarrow \infty .\] 
\label{thm2_emp}
\end{theorem}
\begin{remark}
Theorem \ref{thm1_emp} and Theorem \ref{thm2_emp} hold under misspecification, that is  if the true variance parameter is not same over  different partitions $P^*_{j}$'s. 
\label{diff_var}
\end{remark}

Next, we address the variable selection issue. We have already shown that under any covariate combination, the model with incorrectly selected  change points has Bayes factor converging to zero with respect to the model with true change points and covariate combination. Showing variable selection consistency with the change points set to true change points then boils down to showing the consistency of variable selection. Let $\Pi(M_{l^*,m_\beta}|\cdot)$ be the posterior probability of a model with covariate combinations given $m_\beta$ and change points at true change point locations $t_1^*,\cdots,t_{l^*}^*$, and let $P_n(M_{l^*,m_\beta}:M^*)$ be the ratio of the posterior probabilities  $\Pi(M_{l^*,m_\beta}|\cdot)$ and $\Pi(M^*|\cdot)$, based on $n$ observations. Similarly, we can define $P_n(M_{l,m_\beta}:M^*)$.  We have the following result regarding variable selection.

\begin{theorem}
Under $A1-A7$,  for \eqref{model2}, $ BF(M_{l^*,m_\beta},M^*)\rightarrow 0$ in probability for $m_\beta\neq m_{\beta^*}$.
\label{vr_selection1}
\end{theorem}

For showing consistency over all possible variable choices we assume the following.
\begin{itemize}
\item[(B1)]$-\log P(I_m^\beta=1)\sim (\log n)^{1+\alpha_1}$, for any $\alpha_1>0$.
\item[(B2)] $\sqrt{n} \succ q_n^{2.5}\log n$ and $\sqrt{n\epsilon_n} \succ q_n^{1.5}\log n$.
\end{itemize}

\begin{remark}
Condition $B1$ imposes a stronger penalty on the larger model, which induces selection consistency uniformly over covariate choice and model size. Similarly $B2$ is needed to bound the Bayes factors uniformly over covariate choices.

\end{remark}

We consider the models $M^{\epsilon_n}_{l^*,m_\beta}$ defined as earlier. Then we have the following. 
\begin{theorem}
Under the assumptions $A1-A7$ and $B1, B2$, for \eqref{model2}, $\sup \sum_{m_\beta \neq m_\beta^*}P_n(M^{\epsilon_n}_{l^*,m_\beta} : M^*)\rightarrow 0$ in probability, where the supremum is over the possible  change point selections  and $\epsilon_n$ converges to zero.
\label{bfctr_vrsel}
\end{theorem}

For any model with change points $t_1<\cdots<t_l$ we can state the following result for variable selection. 

\begin{theorem}
Let $M_{l,m_\beta}$ be the corresponding model with change points at $t_1<\cdots<t_l$ and some covariate combinations $m_{\beta}$. Under $A1-A7, B1,B2$,  for \eqref{model2}, $\sum_{m_\beta:m_\beta\neq m_{\beta^*}}P_n(M_{l,m_\beta}:M^*) \to 0$ in probability. \label{vr_selection2}
\end{theorem}

\begin{remark}
Theorems \ref{bfctr_vrsel} and Theorem \ref{vr_selection2} hold under unknown $\sigma^2$, for the empirical estimator $\sigma^2$ as in Theorem \ref{thm1_emp} and \ref{thm2_emp}, and under misspecification of equal variance as in Remark \ref{diff_var}. 
\label{diff_var2}
\end{remark}

\begin{remark}
{\bf Covariate free cases.}
One special case of the model given in \eqref{model1} is the covariate free cases, which is the simple mean model. For such model the result given in Theorem  \ref{thm2}  regarding Bayes factor consistency holds for $\epsilon_n$ going to zero for $\epsilon_n \succeq n^{-1} (\log n)^2$, which  gives  us rate equivalent to frequentist minimax rate up to logarithmic factors. 
\end{remark}

\section{Simulation}\label{sec:simulation}

In this section, we demonstrate the performance of the Bayesian hierarchical model described in \eqref{model1} and \eqref{model2}, for changing linear model and the simpler special case of piecewise constant mean model, respectively. First, we show the recovery of true mean as well as the true change-point locations the simple changing mean model, and then we consider a case with number of covariates $d_n \sim n$ case where we show that we can achieve accuracy in both variable selection and change point estimation under the Bayesian model. Our goal here is not to establish superiority of the Bayesian method used here over extant methods, but rather to show that the methods are not just theoretically optimal, they also have satisfactory small sample performance. 

\subsection{Example 1}\label{sim:1}
We consider an example originally reported in \citep{frick2014multiscale} and compared against an empirical Bayes procedure in \citep{martin2017asymptotically}, with $6$ change points for piecewise constant Gaussian sequence model. Here the data-generating model is given as follows:

\beq 
y_i=\theta_i+ \epsilon_i, \quad \epsilon_i \iid \NormRV(0,0.04), \; i = 1, \ldots, n \; ( = 497),
\eeq
with the true mean being: 
\[
\theta_i = \begin{cases}-0.18; \hspace{0.02in}& 1\leq i \leq 138\\
0.08; \hspace{0.02in} &139\leq i \leq 225\\
1.07; \hspace{0.02in} &226 \leq i \leq 242\\
-0.53; \hspace{0.02in}& 243\leq i\leq 299\\
0.16; \hspace{0.02in} &300\leq i\leq 308\\
-0.69; \hspace{0.02in}& 309\leq i\leq 333\\
-0.16; \hspace{0.02in} &334\leq i\leq 497.
\end{cases}
\]

The sequence of true means $\theta_i$'s are depicted in Fig. \ref{frick1}. 

\begin{figure}[ht!]
\centering
\includegraphics[height=2.3in,width=3.5in]{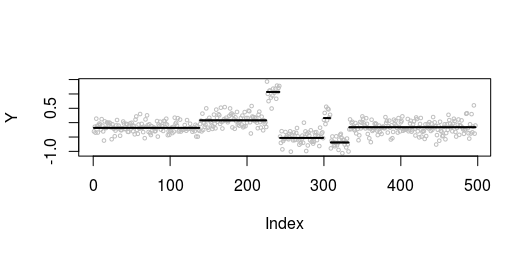}
\caption{True mean parameter values used in \citep{frick2014multiscale}}
\label{frick1}
\end{figure}

We compare the recovery and estimation performance of the method proposed here with two candidates: the first is a frequentist method (the pruned exact linear time method or PELT, \citep{killick2012optimal}) and the second is an empirical Bayes approach from \citep{martin2017asymptotically} (EB). We describe these two comparative candidates briefly. 

For PELT, consider ordered data-points: $y_1, y_2, \ldots, y_n$, and $m$ change-points $\tau_1, \ldots,\tau_m$ that divide the data into $m+1$ partitions. The change-point detection methods then seek to minimize a function:
\[
\sum_{i=1}^{m} C(y_{\tau_{i-1}+1:\tau_i})+ \text{pen}(n) f(m),
\]
where $C$ is a cost-function and $\text{pen}(n) \times f(m)$ is the penalty applied to prevent over-fitting. 
For observations $y_1, \ldots, y_n \sim f(y \mid \theta)$ for some unknown underlying parameter $\theta$. The PELT method uses the negative log-likelihood as the cost function:$ C(y_{(t+1):s}) = - \max_{\theta} \sum_{i=t+1}^{s} f(y_i \mid \theta)$. The penalty is chosen based on the inferential goal, e.g. $\text{pen}(n) = n \log(n)$ is the popular BIC penalty and $f(m) = m$ assumes that penalization is linear with the number of change-points. When $m$ is not too large, BIC favors a parsimonious model and can be shown to be model selection consistent. 

The empirical Bayes (EB) approach in \citep{martin2017asymptotically} works via specification of priors on block-specific parameter vectors ($\theta_B$) and block configurations, where the prior centers on mean parameters are data-dependent. In particular, the mean parameter in each block is assumed to be Gaussian centered on maximum likelihood estimates based on observations in that block, and the block-configuration follows a discrete uniformly distributed partition points, with the configuration size or number of blocks following a truncated geometric distribution. 

For the results shown below, we do not assume known $\sigma^2$ and assume that they can be different over partitions. We use  $V=1$, $I_i \sim \text{Bernoulli}(1/n)$, and calculate the posterior summaries based on $8,000$ Markov chain monte carlo samples with first $4,000$ burn-ins. The fitted mean and the posterior probabilities for partitions are given in Fig. \ref{fitex1}. 

\begin{figure}[ht!]
\centering
\begin{subfigure}{0.7\linewidth}
\includegraphics[width = \textwidth]{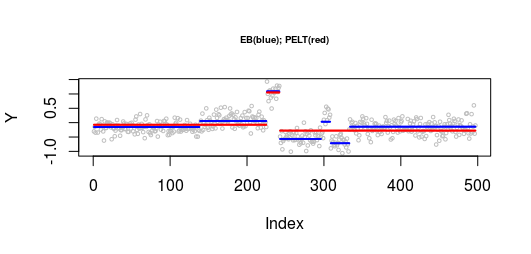}
\end{subfigure}
\vspace{-0.15in}
\begin{subfigure}{0.31\linewidth}
\includegraphics[width = \textwidth]{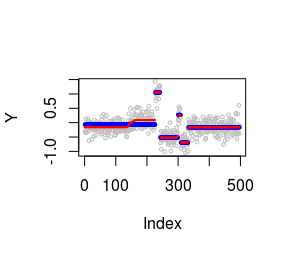}
\end{subfigure}
\begin{subfigure}{0.32\linewidth}
\includegraphics[width = \textwidth]{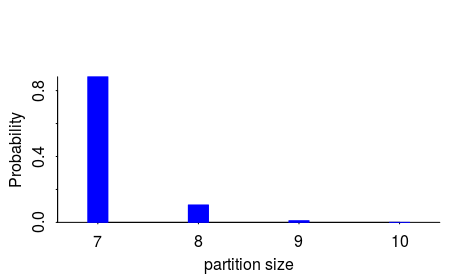}
\end{subfigure}
\begin{subfigure}{0.32\linewidth}
\includegraphics[width = \textwidth]{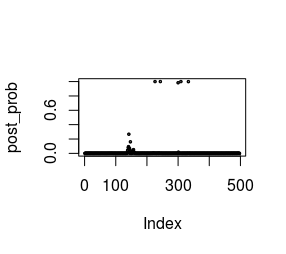}
\end{subfigure}
\caption{{Top row: change-point recovery by EB/PELT; Bottom row: Lower left: posterior mean(red) and mean based on estimated change-points (blue); middle: posterior probability of the number of partitions; Lower right: posterior probability for change points at $i=1,\dots,n$.}}
\label{fitex1}
\end{figure}

%
%
%
%
%


\subsection{Example 2: a case with covariates}\label{sim:2}

Consider a changing linear regression problem where the underlying linear model changes between different observation windows or epochs. Here the parameters of interest are both the parameter vector $\bbeta$ as well as the number and location of change points. Let us fix the dimensions of observations and covariates to be $n=250, p=250$. Suppose the true locations of change-points as a fraction of the total number of observations are given by : $t_1^*=0.3, t_2^*=0.7$. Finally, let the covariates for the $i^{th}$ observation $x_{i,1}\ldots,x_{i,p}$ ($i = 1, \ldots, n$), are generated from independent standard normal distribution. The data-generating model used for this experiment is given as follows: 

\[
y_i = \begin{cases}
                       3+x_{i,2}+2x_{i,12}+1.2e_i & \text{ if } i\leq 75 \\
											 1+2x_{i,2}+.8e_i& \text{ if } 75<i\leq 175\\
											 -2.5+2x_{i,2}-x_{i,3}+e_i  & \text{ if } 175<i\leq 250.
											\end{cases}, i = 1, \ldots, 250. 
\]
Here, $e_i \iid \NormRV(0,1)$ and the true change points occur in positions $i = 75$ and $i = 175$ as mentioned before, and {the proportion of change-points in observations is $p_n=\frac{1}{n}$}. 

\begin{figure}[ht!]
\centering
\begin{subfigure}{0.45\linewidth}
\includegraphics[width=\textwidth]{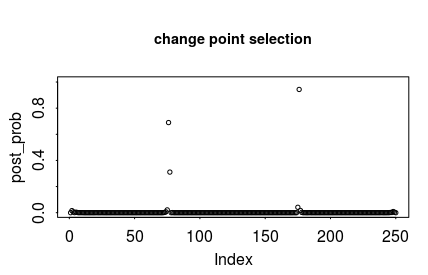}
\caption{Change point recovery}
\label{fig:3a}
\end{subfigure}
\begin{subfigure}{0.45\linewidth}
\includegraphics[width=\textwidth]{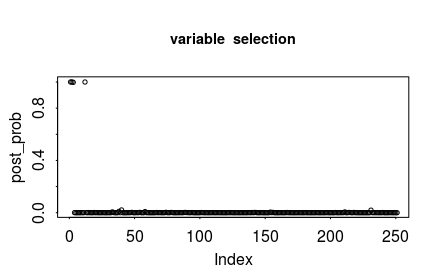}
\caption{Globally active variable selection performance }
\label{fig:3b}
\end{subfigure}
\begin{subfigure}{0.5\linewidth}
\includegraphics[width=\textwidth]{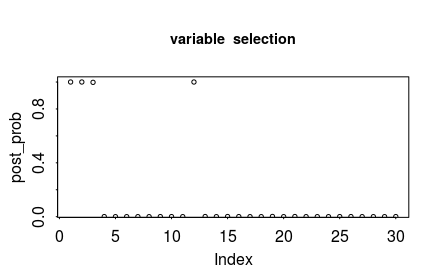}
\caption{Globally active variable selection performance}
\label{fig:3c}
\end{subfigure}
\caption{Changing linear regression example showing performance of the hierarchical Bayesian model}
\label{fig:reg1}
\end{figure}

We use a spike-slab prior on the regression coefficients to detect the global set of covariates and the change points. We let the Markov chain Monte Carlo chain run for $8,000$ iterations and calculate the posterior modes and means for the two indicator variables for the change-points and the non-null $\beta_j$s. The posterior mean estimates are plotted in Fig. \ref{fig:reg1}. Figure \ref{fig:3a} shows that the both the two change-points at $i/n = 0.3$ and $0.7$ can be recovered with high posterior probability. Figures \ref{fig:3b} and \ref{fig:3c} show that the global set of active $\beta_j$'s can also be recovered with high probability.  In particular, from Fig. \ref{fig:reg2},  almost 93\% of the posterior  samples give the correct number of partitions (3), and almost 90\% of the posterior samples select the right model. Thus, the numerical results are in concurrence with our theoretical proofs of consistency in model selection and change-point detection for a changing linear regression in \S \ref{sec:theory}.

\begin{figure}[ht!]
\centering
\begin{subfigure}{0.45\linewidth}
\includegraphics[width=\textwidth]{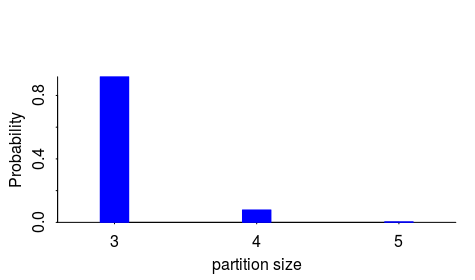}
\caption{Posterior distribution for number of change-points}
\end{subfigure}
\begin{subfigure}{0.45\linewidth}
\includegraphics[width=\textwidth]{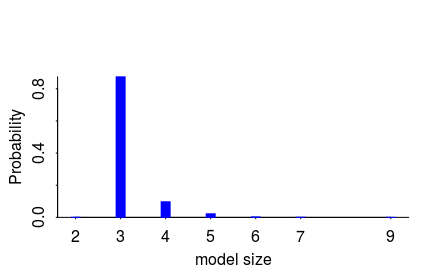}
\caption{Posterior distribution for number of active covariates}
\end{subfigure}
\caption{Regression example: posterior distribution of model size and partition size.}
\label{fig:reg2}
\end{figure}

\subsection{Example 3: a case with covariates and time dependent component}

This example considers a change in linear structure when an autoregressive time dependent component is present. Such scenario may arise in economic application, when for example housing price index may change with the covariates stock market return, but an autoregressive structure may be present for the response variable, i.e. price.  We use a similar model as before with $n=300$, $p=250$ but add an autoregressive component of first order (AR(1)) with autocorrelation equal to $\rho$. As before, the true change-points occur at positions $i = 90$ and $i = 210$, at relative positions $t^* = 0.3$ and $t^* = 0.7$ as before. We generate data from the following model: 
 \[
y_i = \begin{cases}
                       3+\rho y_{i-1}+x_{i,2}+3x_{i,12}+1.2e_i & \text{ if } i\leq 90 \\
											 1+\rho y_{i-1}+2x_{i,2}+.8e_i& \text{ if } 90<i\leq 210\\
											 -2+\rho y_{i-1}+x_{i,2}-x_{i,3}+e_i  & \text{ if } 210<i\leq 300,
											\end{cases}
\]
with $\rho=0.5$, $x_{i,j}$, $e_i$'s are i.i.d. $\NormRV(0,1)$. It should be noted that our theoretical results from \S \ref{sec:theory} will continue to hold for this situation, as long as condition \textbf{(A6)} holds for the new design matrix, accounting for the autoregressive structure. 

We use the spike and slab prior on the coefficients and the parameter $\rho$, and use a computational scheme similar to the example in \S \ref{sim:2}. The posterior probabilities for the variable selection and change-point detection are given in Fig. \ref{fig:ar1}, where it can be seen that the true change point locations (\textit{vide} Fig. \ref{fig:ar1a}) and the active variables (\textit{vide} Fig. \ref{fig:ar1b} and Fig. \ref{fig:ar1c}) are selected with high probability. We also note that the estimated posterior inclusion probability for the AR(1) component is $1$.  

\begin{figure}[ht!]
\centering
\begin{subfigure}{0.45\linewidth}
\includegraphics[width=\textwidth]{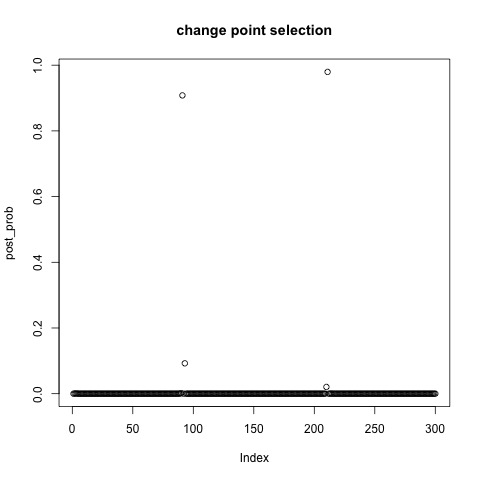}
\caption{Change point recovery}
\label{fig:ar1a}
\end{subfigure}
\begin{subfigure}{0.45\linewidth}
\includegraphics[width=\textwidth]{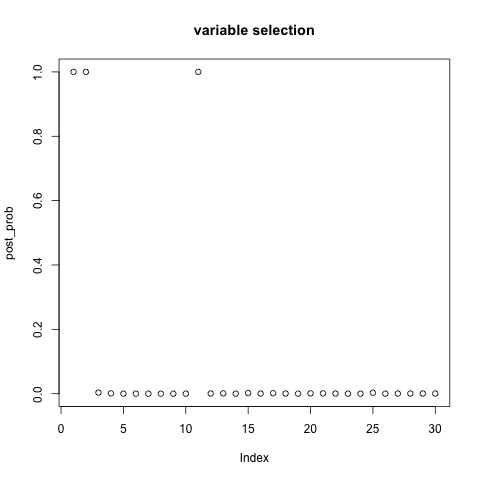}
\caption{Globally active variable selection performance}
\label{fig:ar1b}
\end{subfigure}
\begin{subfigure}{0.5\linewidth}
\includegraphics[width=\textwidth]{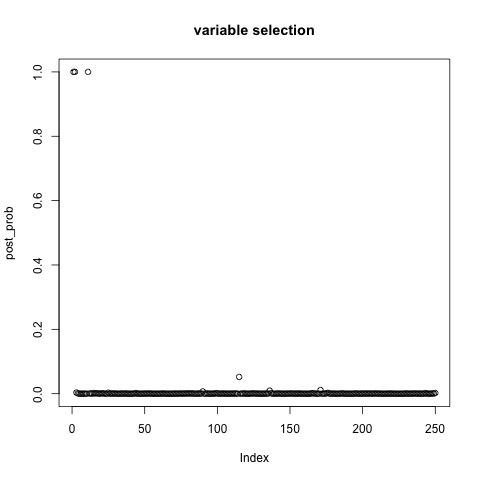}
\caption{Globally active variable selection performance}
\label{fig:ar1c}
\end{subfigure}
\caption{Changing linear regression example showing performance of the hierarchical Bayesian model, in the presence of the  AR(1) component}
\label{fig:ar1}
\end{figure}

\section{Real Data Applications}
Next the proposed method is applied  to detecting change points in crime data from Little Rock, AR. Among all cities in the United States with at least $100,000$ residents, Little Rock is ranked in the top $10$ for the highest violent crime (7th) and property crime (4th) rates in 2015 \citep{chillar2020unpacking}.

For a piecewise constant mean model, there are competing methods against which the proposed method will be compared, as mentioned earlier there is essentially no comprehensive framework for model selection under change point. The data is preprocessed by a square root transformation and standardization. 

For comparing and contrasting the proposed method on the weekly burglary and breaking and entering data from Little Rock from  2017, the empirical Bayes method proposed by \citet{martin2017asymptotically}, as well as the PELT (Pruned Exact Linear Timing) method by \citep{killick2012optimal}, are used. As it is impossible to know if there should be any ``true" change-points in 2017 data in the absence of additional information, this can be regarded as a preliminary exploratory analysis. As Fig. \ref{fig:aa-cp} suggests, the change-points recovered by the proposed method mostly agree with those by the PELT method, while the Empirical Bayes method seems to be conservative and does not detect any change points in the data. 

\begin{figure}[ht!]
\centering
\includegraphics[height=2.2in,width=3.9in]{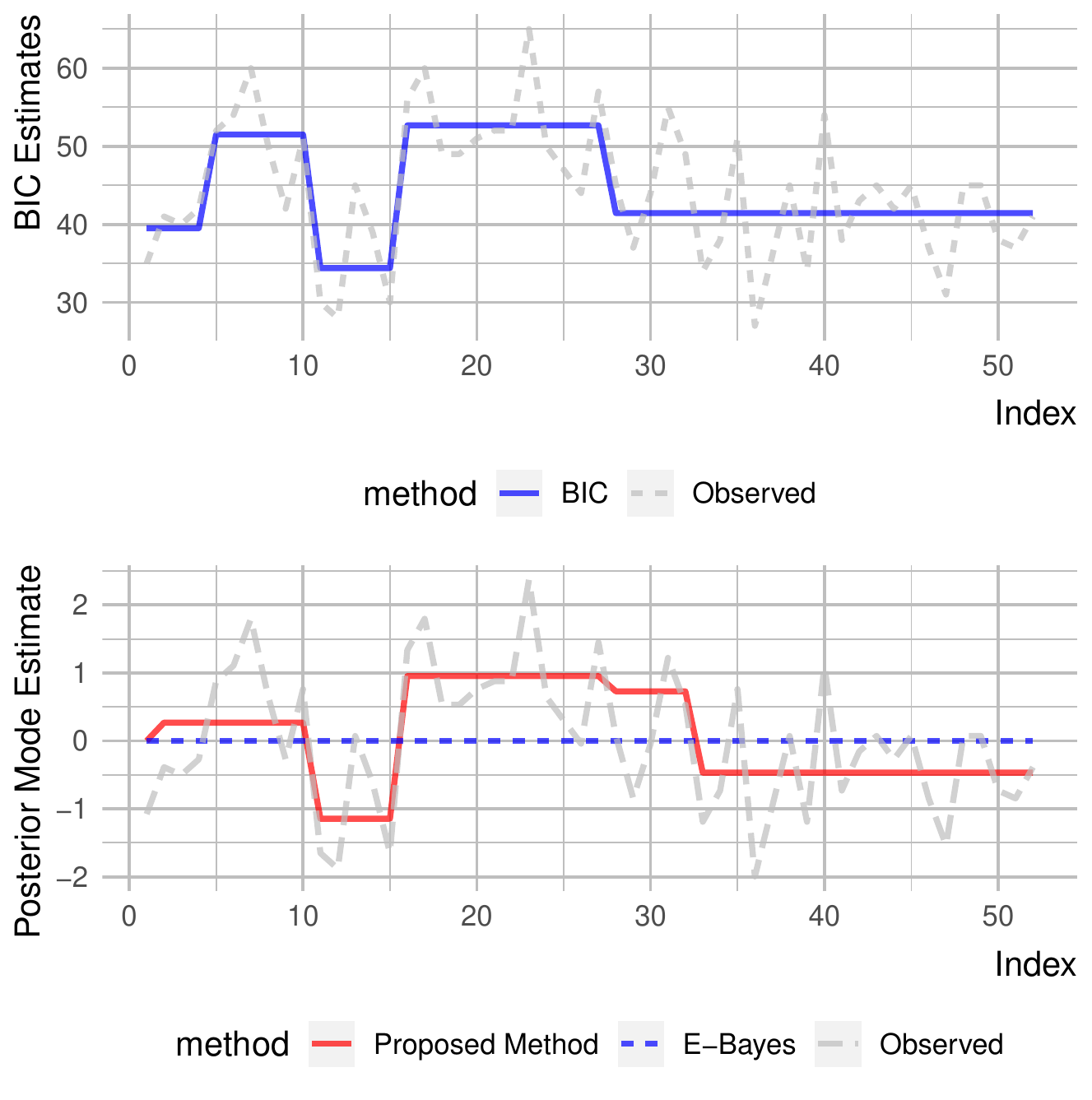}%
\caption{Changepoints detected by the PELT-BIC method (top panel) and the proposed approach and Empirical Bayes approach \citep{martin2017asymptotically} (bottom panel) for the burglary and breaking and entering activities in Little Rock in 2017. Data is standardized for the bottom panel.}%
\label{fig:aa-cp}%
\vspace{-.2in}
\end{figure}


\newpage 

\section{Appendix}\label{sec:appendix}

\section*{Proof of Theorem \ref{thm1}}

We show our result first for a model $M_{l,m_\beta}=M^*_{l,m_\beta}$ with $l$ change points,  where, the $P_j$'s are a refinement of $P^*_{j}$'s, the true partition, and $m_\beta$ contains the true covariate combination, i.e. $m_\beta\supset m_{\beta^*}$.  That is for $t_1,\cdots, t_l$ be the proposed change point, and $\tilde{t}_1,\cdots, \tilde{t}_l$ be the change points corresponding to some refinement of the true partition  $t^*_1,\cdots,t^*_{l*}$ such that $t_i = \tilde{t}_i, \forall i $. Then we show that $BF(M^*_{l,m_\beta},M^*_{l^*,m_\beta})\rightarrow 0$ and  $BF(M^*_{l^*,m_\beta},M^*)\rightarrow 0$. (Case i) 

Then we show the case where $t_1, \ldots, t_l$ be the proposed change point, and $\tilde{t}_1, \ldots, \tilde{t}_l$ be change point corresponding to any refinement of the true partition  $t^*_1,\ldots,t^*_{l*}$ and $\max_i|t_i-\tilde{t}_i|>\epsilon>0$ (Case ii). Without loss of generality, we assume $\epsilon< \min_i|t_i-t_{i-1}|; i=1, \ldots, l$. 

Then we show the case where the $m_\beta$ does not contain the true covariate combination for some partition for $l \geq l^*$. (Case iii) 

Finally we show for the case where $l<l^*$. (Case iv)

\subsection*{ \underline{\bf Case (i) }}

We show  $BF(M^*_{l,m_\beta},M^*_{l^*,m_\beta}) \to 0$ and $BF(M^*_{l^*,m_\beta},M^*) \to 0$ in probability. 

\underline{$BF(M^*_{l,m_\beta},M^*_{l^*,m_\beta}) \to 0$}

Writing the $P^*_{j}$ for some $j$ for the true partition as the union of $P^*_{j_1}, \cdots, P_{j^*_k}$, $k>1$, where $P^*_{j_i}=P_{j'}$ for some $j'\in \{1,\cdots,l+1\}$. Let ${\bf Y}_n^{j^*}$ be the vector of $y_i$'s in $P^*_j$, and $X_{j^*}$ be the corresponding covariate matrix.  Similarly, $X_{j_i^*}$ be the covariate matrix for $P^*_{j_i}$. Let $\hat{\bf Y}_n^{j^*}$ and $\hat{\bf Y}_n^{j_i^*}$ be their least square fit based on $P^*_j$ and $P^*_{j_i}$'s. Let $\hat{{\bf Y}}^{j^*}_{n,i}$ be the sub-vector  of $\hat{{\bf Y}}^{j^*}_{n}$ corresponding to the observations in $P^*_{j_i}$. 

Now 
\begin{align*}
-\log L({\bf Y}^{j^*}_n \mid M^*_{l,m_\beta}) +\log L ({\bf Y}^{j^*}_n \mid M^*_{l^*,m_\beta}) & =
-\frac{1}{2} {\rm det}(X_{j^*}'X_{j^*})+\frac{1}{2}\sum_{i=1}^k {\rm det}(X_{j_i^*}'X_{j_i^*})-\frac{1}{2\sigma^2}\ \| {\bf Y}^{j^*}_n-\hat{{\bf Y}}^{j^*}_n \|^2\\ 
& + \sum_{i=1}^k\frac{1}{2\sigma^2}\|{\bf Y}^{j_i^*}_n-\hat{{\bf Y}}^{j_i^*}_n\|+c_{l,l^*}
\end{align*}

where $c_{l,l'}=O(k(l-l'))$, where $k$ is te number of covariates in $m_\beta$. Now, \[ \|{\bf Y}^{j^*}_n-\hat{{\bf Y}}^{j^*}_n\|^2=\sum_{i=1}^k \|{\bf Y}^{j_i^*}_n-\hat{{\bf Y}}^{j^*}_{n,i}\|^2=\sum_{i=1}^k \|{\bf Y}^{j_i^*}_n-\hat{\bf Y}_n^{j_i^*}+\hat{\bf Y}_n^{j_i^*}-\hat{{\bf Y}}^{j^*}_{n,i}\|^2=\sum_{i=1}^k\|{\bf Y}^{j_i^*}_n-\hat{\bf Y}_n^{j_i^*}\|^2+ \|\hat{\bf Y}_n^{j_i^*}-\hat{{\bf Y}}^{j^*}_{n,i}\|^2.\]
The last step follows as the  cross product term is equal to zero, as for $\delta_{i,m}$ the $m$ th component of $(\hat{\beta}^j-\hat{\beta}^{j_i})$, $x_{i',m}$ the $m$ th component of  the covariate vector $x_{i'}$ for  observation $i'$, $\sum_{i'\in P^*_{j_i}}(y_{i'}-x_{i'}\hat{\beta}^{j_i})x_{i'}(\hat{\beta}^j-\hat{\beta}^{j_i})=\sum_{m}\delta_{i,m}\sum_{i'\in P^*_{j_i}} (y_{i'}-x_{i'}\hat{\beta}^{j_i})x_{i',m}=0$.

Next, 
\[ \|\hat{\bf Y}_n^{j_i^*}-\hat{{\bf Y}}^{j^*}_{n,i}\|^2\leq 2[(\hat{\beta}^{j}-{\tilde{\beta}}^{*j})'(X_{j^*_i}'X_{j^*_i})(\hat{\beta}^{j}-{\tilde{\beta}}^{*j})+(\hat{\beta}^{j_i}-{\tilde{\beta}}^{*j})'(X_{j^*_i}'X_{j^*_i})({\hat{\beta}}^{j_i}-{\tilde{\beta}}^{*j})],\]
where ${\tilde{\beta}}^{*j}$ be the true coefficient vector corresponding to $P^*_j$ with zeros in place for covariates that are in $m_\beta$ but not in $m_{\beta^*}$. We have,  $(\hat{\beta}^{j_i}-{\tilde{\beta}}^{*j})'(X_{j^*_i}'X_{j^*_i})(\hat{\beta}^{j_i}-{\tilde{\beta}}^{*j})\sim \chi^2_{q}$ for $q=\#\{m_\beta\}$  and $q\leq q_n$. Again, $ (\hat{\beta}^{j}-{\tilde{\beta}}^{*j})'(X_{j^*_i}'X_{j^*_i})(\hat{\beta}^{j}-{\tilde{\beta}}^{*j})=Z'A^{1/2}(X_{j^*_i}'X_{j^*_i})A^{1/2}Z$ where $Z\sim N({\bf 0},\sigma^2I_q)$ and $A^{-1}=X_{j^*}'X_{j^*}$. Given the eigenvalues of $A^{1/2}(X_{j^*_i}'X_{j^*_i})A^{1/2}$ are bounded away from infinity, we have   $ (\hat{\beta}^{j}-{\beta^*}^j)'(X_{j^*_i}'X_{j^*_i})(\hat{\beta}^{j}-{\beta^*}^j)\sim O_p(q_n)$.

Next, we consider the determinant term. Let $n_j^*$ be the number of observations in $P_j^*$. Similarly we define $n^*_{j_i}$ for $P^*_{j_i}$'s, and $\sum_{i}n^*_{j_i}=n_j^*$, and $\text{liminf } \frac{n_{j_i}^*}{n} >0$. Let $lim_{n \uparrow\infty} \frac{n_{j_i}^*}{n ^*_j}=\alpha_i$ . From the fact that $\log det(X_{j_i^*}'X_{j_i^*})=q[\log n^*_{j} +\log \alpha_i]+e_q$, for large $n$ where $|e_q|\leq q[ |log a|+|\log b|]$, we have  
 \[-\log({\bf Y}^{j^*}_n|M^*_{l,m_\beta})+\log ({\bf Y}^{j^*}_n|M^*_{l^*,m_\beta})\geq  \frac{1}{2} (k-1) q\log n-qO_p(1)+C_0\]
for a generic constant $C_0$ and hence, $BF(M^*_{l,m_\beta},M^*_{l^*,m_\beta})\rightarrow 0$ in probability. 

\underline{$BF(M^*_{l^*,m_\beta},M^*)\rightarrow 0$}

For each partition $P^*_j$, \[2\log L({\bf Y}^{j^*}_n|M^*)-2\log L ({\bf Y}^{j^*}_n|M^*_{l^*,m_\beta})\sim  (q-\#\{m_{\beta^*}\}) \log n +\frac{1}{\sigma^2}\chi^2_{(q-\#\{m_{\beta^*}\})}+(q-\#\{m_{\beta^*}\})C'_0\]for a generic constant $C'_0$, where $q=\#\{m_\beta\}$. Hence, the result follows.

\subsection*{ \underline{\bf Case (ii) }}

From \eqref{decomp0}, we have $E_1=\sum_jE_{1,j}$, $E_2=\sum_jE_{2,j}$, and $E_3=\sum_jE_{3,j}$, where for $e_i=y_i-\theta^*_i$,
\begin{eqnarray}
E_{1,j} &= &\sum_{i \in P_j}e_i(x_iE[\hat{\beta}^j] -x_i\tilde{\beta}^{*}_i); \; E_{2,j}=\sum_{i \in P_j}e_i(x_i\hat{\beta}^j -x_iE[\hat{\beta}^{j}]); \nonumber \\
E_{3,j} &=& \sum_{i\in P_j}(x_iE[\hat{\beta}^j] -x_i\tilde{\beta}^{*}_i)(x_i\hat{\beta}^j -x_iE[\hat{\beta}^{j}])
\label{e_rate}
\end{eqnarray}

Let, $\delta^{n,i}_{m}$ be the $m$ th component of $E[\hat{\beta}^j] -\tilde{\beta}^{*}_i$, and $\delta^{n,E}_{m}$ be the $m$ th component of $\hat{\beta}^j-E[\hat{\beta}^j]$. Note that $\delta^{n,i}_{m}$ can take $l^*+1$ many different values, as $i \in P^*_j$ for some  $j\in\{1,\cdots, l^*+1\}$. We denote it by $\delta^n_{m,l'}$ for $i \in P_j\cap P^*_{l'}$.

Note that $E[\hat{\beta}^j]$ is linear combination of ${\tilde{\beta}}^{*j}$'s of the form $(\sum_{i\in I_j} B_i'B_i)^{-1}(\sum_{i\in I_j} B_i'B_i{\tilde{\beta}}^{*i})$, where $I_j$ is a subset of $\{1,\cdots,l^*+1\}$, and $i\in I_j$ if $P^*_i\cap P_j\neq \phi$ for some $i$, and $B_i$ be the corresponding covariate matrix for observations in $P^*_i\cap P_j$. Hence,  $E[\hat{\beta}^j]$, bounded at each component by  condition A6. For $E_{2,j}$, $E_{3,j}$ we use the fact  that $n$ dimensional multivariate normal with bounded variance, the absolute value maximum is bounded by $\log n$ for large $n$. 
Then,
\[|E_{1,j}|=|\sum_{m=1}^q\sum_{l'=1}^{l^*+1}\delta^n_{m,l'}\sum_{i \in P_j \cap P_{l'}^*}e_ix_{i,m}|=o_p(q\sqrt{n}\log n),\]
\[|E_{2,j}|=|\sum_{m=1}^q (\sqrt{n}\delta^{n,E}_{m})\frac{1}{\sqrt{n}}\sum_{i \in P_j}e_ix_{i,m}|=o_p(q\log n),\]
\[|E_{3,j}|=|\frac{1}{\sqrt{n}}\sum_{l'=1}^{l^*+1}\sum_{i\in P_j\cap P^*_{l'}}\sum_{m,m'}\delta^{n}_{m,l'}(\sqrt{n}\delta^{n,E}_{m'})x_{i,m}x_{i,m'}|=o_p(\sqrt{n}q^2\log n).\]

Without loss of generality we assume that $l=l^*$ (otherwise, we can show for a refinement of true partition for change points
 $\tilde{t}_1,\cdots, \tilde{t}_l$, such that $|t_j-\tilde{t}_j|>0$ and $\tilde{t}_j\in \{t^*_1,\cdots,t^*_{l^*}\}$, then use the result proved in Case (i)).

Let $|t_j-t^*_j|>\epsilon$ and the model be denoted by $M^\epsilon_{l^*,m_\beta}$. Then,

\begin{eqnarray}
-\log L({\bf Y}_n|M^\epsilon_{l^*,m_\beta})+ \log L({\bf Y}_n|M^*_{l^*,m_\beta}) =-\frac{1}{2}[\sum_j\log(det(X_{j^*}'X_{j^*}))-\sum_j\log(det(X_{j}'X_{j}))]+&\nonumber \\
 \frac{1}{2\sigma^2}\sum_{j:P_j\cap P_j^*\neq \phi}({\tilde{\beta}}^{*j}-E[\hat{\beta}^j])'X_{j\cap j^*}'X_{j\cap j^*}({\tilde{\beta}}^{*j}-E[\hat{\beta}^j])+&\nonumber \\ \frac{1}{2\sigma^2}\sum_j\sum_{k:k\neq j, P_j\cap P_k^*\neq \phi}({\tilde{\beta}}^{*k}-E[\hat{\beta}^j])'X'_{j\cap k^*}X_{j\cap k^*}(\tilde{\beta}^{*k}-E[\hat{\beta}^j])+&\nonumber\\
\frac{1}{2\sigma^2}\sum_j(\hat{\beta}^j-\E[\hat{\beta}^j])'X_j'X_j(\hat{\beta}^j-E[\hat{\beta}^j])-&\nonumber\\\frac{1}{2\sigma^2}\sum_j(\hat{\beta}^{*j}-\E[\hat{\beta}^{*j}])'X_{j*}'X_{j*}(\hat{\beta}^{*j}-E[\hat{\beta}^{*j}])-R_n&\nonumber\\
\label{decomp11}
\end{eqnarray}
where $R_n=o_p(\log n \sqrt {n} q_n^2)$, and $\sum_j(\hat{\beta}^{*j}-\E[\hat{\beta}^{*j}])'X_{j*}'X_{j*}(\hat{\beta}^{*j}-E[\hat{\beta}^{*j}])\sim \chi^2_{ql^*}$.

As, $\beta^{*j}\neq \beta^{*{(j+1)}}$ for $j=1,\cdots l^*$. Therefore,  one of the first two quadratic form sums is computed for a nonzero vector for some $j$. The first quadratic is based on $\sim n(1-\epsilon)$ observations and the second one is based on $\sim n\epsilon$ observations, and therefore by A6
\[-\log L({\bf Y}_n|M^{\epsilon_n}_{l^*,m_\beta})+ \log L({\bf Y}_n|M^*_{l^*,m_\beta})\succeq n\epsilon-R_n\rightarrow \infty\]
which proves our claim. 

\subsection*{\underline{\bf Case (iii)}}

Case (iii) follows similar to last step in Case (ii), as the bounds on $E_1, E_2, E_3$  are of same order as the earlier part, and we have $\E[\hat{\beta}^j]\neq {\tilde{\beta}}^{*j}$ for some $j$.

\subsection*{\underline{\bf Case (iv)}}

 We have
\begin{eqnarray*}
P({\bf Y}_n \mid M_{l,m_\beta})  =-\frac{1}{2}\sum_j\log(\abs{X_j'X_j})-\frac{1}{2\sigma^2}[\sum_{i}(y_i-\theta^*_i)^2+\\ \sum_j\sum_{j_1:P_j\cap P_{j_1}^*\neq \phi}({\tilde{\beta}}^{*j_1}-E[\hat{\beta}^j])'X_{j\cap{j_1}^*}'X_{j\cap{j_1}^*}({\tilde{\beta}}^{*j_1}-E[\hat{\beta}^j])\\  +\sum_j(\hat{\beta}^j-\E[\hat{\beta}^j])'X_j'X_j(\hat{\beta}^j-E[\hat{\beta}^j]) ]
 -\frac{1}{\sigma^2}[E_1+E_2+E_3]+c_{l,n}+O(1) .
\end{eqnarray*}
where  $X_{j\cap{j_1}^*}$ is the design matrix corresponding to $P_j\cap P_{j_1}^*$ if $P_j\cap P^*_{j_1}\neq \phi$, $j_1\in=\{1,\dots,l^*+1\}$.
For $l<l^*$, we have $min\{\mathscr{L}([t_{i-1},t_i]\cap[t^*_{j-1},t^*_j]),\mathscr{L}([t_{i-1},t_i]\cap[t^*_j,t^*_{j+1}])\}>0$ for some $i,j$,  if $t_1,\dots,t_l$ correspond to the change points in $M_{l,m_\beta}$, and $\mathscr{L}(\cdot)$ is the Lebesgue measure. Then, $\sum_j\sum_{j_1:P_j\cap P_{j_1}^*\neq \phi}({\tilde{\beta}}^{*j_1}-E[\hat{\beta}^j])'X_{j\cap{j_1}^*}'X_{j\cap{j_1}^*}({\tilde{\beta}}^{*j_1}-E[\hat{\beta}^j])\sim O(n)$. The bounds on $E_1,E_2,E_3$ from Case (ii) hold and hence,  $-P({\bf Y}_n \mid M_{l,m_\beta})+P({\bf Y}_n \mid M_{l^*,m_\beta})=O_p(n)$, which proves our claim.

\section*{Proof of Proposition \ref{prop1}}

Marginalizing over the coefficient vector on $P_j$ we have, 
\begin{eqnarray}
\log L({\bf Y}_n^j|M_{l,m_\beta})&=&-\frac{n_j}{2}\log \sigma^2-\log det(X_j'X_j+S_\beta)+\log det(S_\beta) \nonumber \\
&& -\frac{1}{2\sigma^2}[{{\bf Y}_n^j}'{\bf Y}_n^j-{{\bf Y}_n^j}'X_j'(X_j'X_j+S_\beta)^{-1}X_j{\bf Y}_n^j] \label{marcov2}
\end{eqnarray}
where $\sigma^2S^{-1}_\beta$ is prior variance covariance matrix for the coefficients in $m_\beta$.

Next, we consider $(X_j'X_j+S_\beta)^{-1}$. Let $AA=X_j'X_j$, where $A$ is a positive definite matrix with eigenvalues of the order of $\sqrt{n_j}$.

Then, using the Neumann series expansion \citep[][p.348]{horn2012matrix}, we arrive at: 
\begin{eqnarray}
(X_j'X_j+S_\beta)^{-1}=A^{-1}(I_q+A^{-1}S_\beta A^{-1})^{-1}A^{-1}=A^{-1}(I_q-B-B^2-B^3-\cdots)A^{-1}
\label{mat_inf}
\end{eqnarray}
where $B=A^{-1}S_\beta A^{-1}$ has Eigen values of the order $n^{-1}$ and the above expression is valid for sufficiently large $n$. 

Hence,
\begin{eqnarray}
{{\bf Y}_n^j}'X_j'(X_j'X_j+S_\beta)^{-1}X_j{\bf Y}_n^j= {{\bf Y}_n^j}'X_j'A^{-1}A^{-1}X_j{\bf Y}_n^j-\sum_{k=1}^\infty{{\bf Y}_n^j}'X_j'A^{-1}B^kA^{-1}X_j{\bf Y}_n^j.
\label{res_inf}
\end{eqnarray}
Note that $\|{\bf Y}_n^j\|^2\leq 2[\|\underline{\theta}^n\|^2+\|\underline{e}^n\|^2] \preceq n $ with probability one, $\underline{\theta}^n$ is the vector of $\theta_i$'s and $\underline{e}^n$ the vectors of $e_i$'s. The Eigen values of $A^{-1}B^kA^{-1}$ is of the order of $n^{-k-1}$.

Hence, $\|{{\bf Y}_n^j}'X_j'A^{-1}B^kA^{-1}X_j{\bf Y}_n^j\|\sim n^2 n^{-k-1}$, and therefore, $\sum_{k=1}^\infty{{\bf Y}_n^j}'X_j'A^{-1}B^kA^{-1}X_j{\bf Y}_n^j$ is $O(1)$ with probability one. 

Let $\frac{\lambda_i}{n_j}$ be the Eigen values of $B$ for $i=1,\cdots,q$, where $\lambda_i>0$ and bounded. Again, $\log det((X_j'X_j+S_\beta))=\log (det(A))^2+\log  (det(I+B)^{-1})=\log det(X_j'X_j)-\sum _{i=1}^q\log(1+\frac{\lambda_i}{n_j})=\log det(X_j'X_j)-O(\frac{q}{n})$.

Hence, combining the calculation of the determinant and the residual calculation from equation \ref{res_inf}, the result follows.

\section*{Proof of Theorem \ref{thm2}}
We assume that $m_\beta \supseteq m_{\beta^*}$. For the case, where $m_\beta$ does not contain the true covariates, the proof will follow similar to the proof of Case iii of Theorem \ref{thm1}. The proof for $m_\beta \supseteq m_{\beta^*}$ is given as the following. 

From equation \eqref{e_rate}, we decompose $E_{1,j}$ and $E_{2,j}$ in two parts $E_{1,j\cap j^*}$, $E_{1,j-j^*}$, and  $E_{3,j\cap j^*}$, $E_{3,j-j^*}$, where 
$E_{1,j\cap j^*}=\sum_{i \in P_j\cap P^*_j}e_i(x_iE[\hat{\beta}^j] -x_i\tilde{\beta}^{*j})$, and $E_{1,j- j^*}=\sum_{i \in P_j- P^*_j}e_i(x_iE[\hat{\beta}^j] -x_i\tilde{\beta}^{*j})$. Similarly, $E_{3,j\cap j^*}$, $E_{3,j-j^*}$ are defined. 

As in the proof of Theorem \ref{thm1}, $\delta^{n,i}_{m}$ be the $m$ th component of $E[\hat{\beta}^j] -\tilde{\beta}^{*}_i$, and $\delta^{n,E}_{m}$ be the $m$ th component of $\hat{\beta}^j-E[\hat{\beta}^j]$, and $\delta^{n,i}_{m}$ can take $l^*+1$ many different values, as $i \in P^*_j$ for some  $j\in\{1,\cdots, l^*+1\}$. It is denoted  by $\delta^n_{m,l'}$ for $i \in P_j\cap P^*_{l'}$.

Note that, number of observations in $P_j- P^*_j$ is $\preceq n\epsilon_n$ and $\delta^{n,i}_{m}$ in $P_j\cap P_j^*$ is of the order of $\epsilon_n$ (follows from Lemma \ref{lem_bias}).
Hence, 
\[
|E_{1,j\cap j^*}| \leq \sum_m |\delta^n_{m,j}|n^{1/2}n^{-1/2}|\sum_{i \in P_j\cap P^*_j}e_ix_{i,m}|\preceq \epsilon_n n^{1/2}q_n \log n
\]
for large $n$ almost surely. Next, 

\[
|E_{1,j- j^*}|\leq \sum_{l'}\sum_m |\delta^n_{m,l'}|\sqrt{n \epsilon_n}\frac{1}{\sqrt{n \epsilon_n}} |\sum_{i \in P_j\cap P^*_{l'};l'\neq j}e_ix_{i,m}|\preceq q_n \sqrt{n\epsilon_n} \log n
\]
almost surely.

Using, similar calculation, with probability one, 
\[
|E_{3,j\cap j^*}| \leq \sum_{m,m'} |\delta^n_{m,j}|n^{-1/2}|\sum_{i \in P_j\cap P^*_j}x_{i,m'}(n^{1/2}\delta^{n,E}_{m'})x_{i,m}|\preceq \epsilon_n n^{1/2}q^2_n (\log  n),
\]
and,
\[
|E_{3,j- j^*}| \leq\sum_{l'} \sum_{m,m'} |\delta^n_{m,l'}|n^{-1/2}  |\sum_{i \in (P_j- P^*_j)\cap P^*_{l'}:l'\neq j}x_{i,m'}(n^{1/2} \delta^{n,E}_{m'})x_{i,m}|\preceq \epsilon_n n^{1/2}q^2_n (\log  n).
\]

As in the proof of Theorem 3, here we use the fact that the absolute value of the maximum of an $n$-dimensional multivariate normal is less than $\log n$ for large $n$. 

Hence, from equation \eqref{decomp11}, using the stricter bound for $E_1,E_2,E_3$ from the above derivation, 

\begin{eqnarray}
-\log L({\bf Y}_n|M^\epsilon_{l,m_\beta})+\log L({\bf Y}_n|M^*_{l^*,m_\beta})\succeq n\epsilon_n -\sqrt{n}\epsilon_n q_n^2\log n-q_n \sqrt{n\epsilon_n} \log n+q_n\log n\rightarrow \infty\nonumber\\
\label{rate_decomp}
\end{eqnarray}
as $n\rightarrow \infty$, which proves our result, as $BF(M^*_{l^*,m_\beta},M^*)\rightarrow 0$ in probability, as in Theorem \ref{thm1} (Case i,second part), by BIC type quantities for each partition $P^*_j$.

\begin{lemma}
Under the setting of Theorem \ref{thm2}, we have $\|E[\hat{\beta}^j]-{\tilde{\beta}}^{*}_i\|_\infty\preceq \epsilon_n$, for $i \in P_j\cap P_j^*$.
\label{lem_bias}
\end{lemma}
\begin{proof}
Let $Z_j$ be the design matrix corresponding to observations in $P_j\cap P^*_j$, and $Z_{l'}$ be the matrix corresponding to observations in $P_j\cap P^*_{l'}; l'=1,\cdots,l^*+1, j\neq l'$. Note that number of observations corresponding to $P_j\cap P^*_{l'}$ is less than $n\epsilon_n+1$. 

Then,

\[E[\hat{\beta}^j]=(Z_j'Z_j+\sum_{l',l'\neq j}Z_{l'}'Z_l)^{-1}(X_j'E[{\bf Y}_n^j])=(CC+D)^{-1}(Z_j'Z_j\tilde{\beta}^{*j}+\sum_{l',l'\neq j}Z_{l'}'Z_l\tilde{\beta}^{*l'}),\]
where $C$ is a positive definite matrix with Eigen values of the order of $\sqrt{n_j}$, when $\epsilon_n\rightarrow 0$, and $CC=Z_j'Z_j$. 

Now, writing $(CC+D)^{-1}=C^{-1}(I+C^{-1}DC^{-1})^{-1}C^{-1}$ and from the fact the Eigen values of  $B=C^{-1}DC^{-1}$ is of the order $\epsilon_n$, for large $n$ and therefore,  similar to equation \ref{mat_inf} 
\[(CC+D)^{-1}=C^{-1}(I-B-B^2-\cdots)C^{-1}.\] Hence, 
\[E[\hat{\beta}^j]=C^{-1}C^{-1}CC\tilde{\beta}^{*j}+C^{-1}C^{-1}\sum_{l',l'\neq j}Z_{l'}'Z_l'\tilde{\beta}^{*l'}-\sum_{k=1}^\infty C^{-1}B^kC^{-1}(Z_j'Z_j\tilde{\beta}^{*j}+\sum_{l',l'\neq j}Z_{l'}'Z_l'\tilde{\beta}^{*l'}).\]
We have $C^{-1}B^kC^{-1}$ with Eigen values at most of the order of $b^k\epsilon_n^kn^{-1}$. Also, $C^{-1}C^{-1}Z_{l'}'Z_l'$ has Eigen value at most of the order of $\epsilon_n$. We have $\beta^{*j}$ and $\tilde{\beta}^{*j}$ bounded. Hence, \[\|C^{-1}C^{-1}\sum_{l',l'\neq j}Z_{l'}'Z_l'\tilde{\beta}^{*l'}\|_\infty <c^*\epsilon_n; \text { and }\|C^{-1}B^kC^{-1}(Z_j'Z_j\tilde{\beta}^{*j}+\sum_{l',l'\neq j}Z_{l'}'Z_l'\tilde{\beta}^{*l'})\|_\infty \leq c^*b^k(\epsilon_n^k+b\epsilon_n^{k+1}),\] where $c^*>0$ is an universal constant (using A6). 

Hence,  for $i\in P_j\cap P_j^*$, $\|E[\hat{\beta}^j]-\tilde{\beta}^{*}_i\|_\infty=\|E[\hat{\beta}^j]-\tilde{\beta}^{*j}\|_\infty\preceq \epsilon_n$. 

\end{proof}

\section*{Proof of Theorem \ref{thm1_emp} and \ref{thm2_emp}}

We use, $\hat{\sigma}_{M^*}^2=\|{\bf Y}_n-\hat{\bf Y}_n\|^2/n$ for the true model with change points at $t_i^*$, and $\hat{\sigma}_{M_l}^2=\hat{\sigma}_{M_{l,m_\beta}}^2$ be the estimate corresponding to a  model with change point $l$ change point, with covariate combination given by $m_\beta$. 

From earlier calculation in Proposition 1, replacing the $\sigma^2$ in each partition $P_j$ by $\hat{\sigma}^2 $,
\begin{align*}
 \log L({\bf Y}_n|M_{l,m_\beta}) & =-\frac{1}{2}\sum_j\log( det(X_j'X_j))-\frac{1}{2\hat{\sigma}_{M_l}^2}n\hat{\sigma}_{M_l}^2-\sum_j\frac{n_j}{2}\log \hat{\sigma}_{M_l}^2+O(lq_n).
 \end{align*}

 Hence,
\begin{align*}
 \log L({\bf Y}_n|M_{l,m_\beta})&- \log L({\bf Y}_n|M^*)\\
 =&-\frac{1}{2}[\sum_j\log( det(X_j'X_j)-\sum_j\log( det(X_{j^*}'X_{j^*})]-\sum_j\frac{n_j}{2}\log \hat{\sigma}^2_{M^*}(1+\frac{ \hat{\sigma}_{M_l}^2-\hat{\sigma}_{M^*}^2}{\hat{\sigma}^2_{M^*}})+\\
 &\frac{n}{2}\log \hat{\sigma}^2_{M^*}+O(lq_n)\\
 =&-\frac{1}{2}\sum_j\log( det(X_j'X_j)+\frac{1}{2}\sum_j\log( det(X_{j^*}'X_{j^*})-\frac{n}{2}\log (1+\frac{ \hat{\sigma}_{M_l}^2-\hat{\sigma}_{M^*}^2}{\hat{\sigma}^2_{M^*}})+O(lq_n).
  \end{align*}
Note that  $\hat{\sigma}_{M^*}^2 \rightarrow \sigma^2$ in probability, and  $\frac{ \hat{\sigma}_{M^{\epsilon}_{l}}^2-\hat{\sigma}_{M^*}^2}{\hat{\sigma}^2_{M^*}}>0$, $\frac{ \hat{\sigma}_{M^{\epsilon_n}_{l^*}}^2-\hat{\sigma}_{M^*}^2}{\hat{\sigma}^2_{M^*}}>0$ for large $n$ under the set up of Theorem \ref{thm1_emp} and Theorem \ref{thm2_emp}, for $M^\epsilon_{l,m_\beta}$ and $M^{\epsilon_n}_{l^*,m_\beta}$, respectively. (from Theorems \ref{thm1} ,\ref{thm2} proofs). 

For Theorem \ref{thm1_emp},  for large $n$, we have,
\[ \log L({\bf Y}_n|M^\epsilon_{l,m_\beta})- \log L({\bf Y}_n|M^*)\preceq -\frac{1}{2}\sum_j\log( det(X_j'X_j)+\frac{1}{2}\sum_j\log( det(X_{j^*}'X_{j^*})-\frac{n}{2}c_0\frac{ \hat{\sigma}_{M^\epsilon_l}^2-\hat{\sigma}_{M^*}^2}{\hat{\sigma}^2_{M^*}}\preceq-n\epsilon, \]
as $\sigma^2_{M^\epsilon_l}$ is bounded with probability one, and  $\log (1+\frac{ \hat{\sigma}_{M^\epsilon_l}^2-\hat{\sigma}_{M^*}^2}{\hat{\sigma}^2_{M^*}})>c_0 \frac{ \hat{\sigma}_{M^\epsilon_l}^2-\hat{\sigma}_{M^*}^2}{\hat{\sigma}^2_{M^*}}$ for some small positive  constant $c_0$.  Similarly, the proofs  for  the case $m_\beta$ not containing true covariate combination, and the case $l<l^*$ follow.

Under the setup of  \ref{thm1_emp} for a refinement of true partition, and covariate combination containing the true covariates, $\frac{\hat{\sigma}_{M_l}^2-\hat{\sigma}_{M^*}^2}{\hat{\sigma}^2_{M^*}}\rightarrow 0$ and $\log (1+\frac{ \hat{\sigma}_{M_l}^2-\hat{\sigma}_{M^*}^2}{\hat{\sigma}^2_{M^*}})\sim \frac{ \hat{\sigma}_{M_l}^2-\hat{\sigma}_{M^*}^2}{\hat{\sigma}^2_{M^*}}$. Hence, Theorem  \ref{thm1} proof (Case i)   can be repeated, and we have $\log L({\bf Y}_n|M_{l,m_\beta})- \log L({\bf Y}_n|M^*)\preceq -q_n\log n$.

For Theorem \ref{thm2_emp},   we note that $\log (1+\frac{ \hat{\sigma}_{M^{\epsilon_n}_{l^*}}^2-\hat{\sigma}_{M^*}^2}{\hat{\sigma}^2_{M^*}})\succeq  \frac{ \hat{\sigma}_{M^{\epsilon_n}_{l^*}}^2-\hat{\sigma}_{M^*}^2}{\hat{\sigma}^2_{M^*}}$ for large $n$,  as $  \hat{\sigma}_{M^{\epsilon_n}_{l^*}}^2-\hat{\sigma}^2
_{M^*} \geq0$ and is of the order of $ \epsilon_n$ for $m_\beta\supset m_{\beta^*}$,  from the fact $\frac{\log (1+x)}{x} \rightarrow 1$ as $x \rightarrow 0$. Hence,  Theorem \ref{thm2} proof  can be repeated. 

\subsection*{Addressing Remark \ref{diff_var}}

If $\sigma^2=\sigma^2_j$ for $P_j^*$, then $\hat{\sigma}_{M^*}^2$ converges to $ \sum_{j=1}^{l^*}\frac{n_j^*}{n}\sigma^2_j>0$ in probability. Hence, the proofs of Theorem \ref{thm1_emp} and Theorem \ref{thm2_emp} hold. 
\section*{Proof of Theorem \ref{vr_selection2}}
Let $M_{l^*_{t^*},m_\beta}$ (denoted by $M_{l^*,m_\beta}$ for convenience) be the model with covariate combination given by $m_\beta$ and true change points $t_1^*,\cdots,t_l^*$. 
First we show that $\frac{ \sum_{m_\beta:m_\beta\neq m_{\beta^*}}\Pi(M_{l^*,m_\beta}|\cdot)}{\Pi(M^*|\cdot)}\rightarrow 0$ in probability (Part 1). Then, we show that  $\frac{ \sum_{m_\beta:m_\beta\neq m_{\beta^*}}\Pi(M^{\xi}_{l,m_\beta}|\cdot)}{\Pi(M^*|\cdot)}\rightarrow 0$, where $M^{\xi}_{l,m_\beta}$ is a model with change points $t_1,t_2,\cdots,t_l$ such that for the refinements of true partition corresponding to change points,  $\tilde{t}_1,\cdots,\tilde{t}_l$, such that $inf_{\tilde{t}_1,\cdots,\tilde{t}_l}max_i|t_i-\tilde{t_i}|=\xi$ (in Part 2). 
Next, we address the case where $t_1,\dots,t_l$ are change points and $l<l^*$ (Part 3). For the case, $t_1,\cdots,t_l$ is a refinement of $t_1^*,\cdots,t_l^*$, the proof follows from Part 1, by considering $M^*_{l}$ instead of $M^*$(true covariate and the change points $t_1,\dots,t_l$) and concluding $\frac{ \sum_{m_\beta:m_\beta\neq m_{\beta^*}}\Pi(M_{l,m_\beta}|\cdot)}{\Pi(M^*_l|\cdot)}\rightarrow 0$ in probability and from the fact $\frac{\Pi(M^*_l|\cdot)}{\Pi(M^*|\cdot)}\rightarrow 0$ in probability.

\subsection*{\underline{Part 1:$\frac{ \sum_{m_\beta:m_\beta\neq m_{\beta^*}}\Pi(M_{l^*,m_\beta}|\cdot)}{\Pi(M^*|\cdot)}\rightarrow 0$}}

Let $R_{t}^{(j)}$ be the residual sum of square for $P_{j}^*$ under $M^*$. Let $m_\beta\supset m_{\beta^*}$. For $\#\{m_\beta\}=q$ and $\#\{m_{\beta^*}\}=t$. Then we show that for residual sum of square for $m_{\beta^*}$, $R_m^{(j)}$
\[ P(R_t^{(j)}-R_m^{(j)} \geq q-t+  (q-t)\alpha (\log n)^{1+\alpha_1}  ; \text{  for some } m_{\beta}\supset m_{\beta^*}, \{m_\beta\}=q )\leq d_n^{q-t} e^{-c(q-t)\alpha (\log n)^{1+\alpha_1}}  \]
for some universal constant $c>0$, for any $\alpha>0$. Hence,
\begin{align*}
P(R_t^{(j)}-R_m^{(j)} \geq (q-t)+  (q-t)\alpha( \log n)^{1+\alpha_1} ; \text{  for some } m_{\beta} \supset m_{\beta^*}, & \{m_\beta\}=q, \text{for some } q )\\
&\leq \sum_{q=t+1}^{q_n} d_n^{q-t} e^{-c\alpha(q-t)(\log n)^{1+\alpha_1}}.  
\end{align*}

Here, $d_n=p$, the number of available covariates (depending on $n$) and   $\log d_n=O(\log n)$ (by assumption) and  $\delta_n=\sum_{q=t+1}^{q_n} d_n^{(q-t)} e^{-c\alpha(q-t) (\log n)^{1+\alpha_1}}\rightarrow 0$.

Similarly, for any $m_\beta \not \supset m_{\beta*}$, let $R_{m'}^{(j)}$ be the residual sum of square for $m_\beta\cup m_{\beta^*}$.  A conservative bound is given by,
\[ P(R_t^{(j)}-R_{m'}^{(j)} \geq q+  \alpha q( \log n)^{1+\alpha_1}  ; \text{  for some } m_{\beta}\not \supset m_{\beta^*}, \#\{m_\beta\}=q, \text{ for some } q )\leq \sum_{q=1}^{q_n} d_n^q e^{-c\alpha q (\log n)^{1+\alpha_1}}.  \]

For, any $\tilde{m}_\beta$ missing at least one of the true covariates. Let, $R_{m'}^{(j)}$ be the residual sum of square for $\tilde{m}_\beta \cup m_{\beta^*}$, and $R_m^{(j)}$ for $\tilde{m}_\beta$. Then,
\[R_{m}^{(j)}-R_{m'}^{(j)}=(\hat{\beta}^{(j)}_{m'}-\hat{\beta}^{(j)}_m)'X_{j^*}'X_{j*}(\hat{\beta}^{(j)}_{m'}-\hat{\beta}^{(j)}_m),\]
where $\hat{\beta}^{(j)}_{m'},\hat{\beta}^{(j)}_m$ are corresponding least square estimates for $\tilde{m}_\beta \cup m_{\beta^*}$ and $\tilde{m}_\beta$, respectively, with entries corresponding to coefficients not in $m_\beta$ but in $\tilde{m}_\beta$ are filled with zero in $\hat{\beta}^{(j)}_m$, and $X_{j^*}$ is the design matrix for $P_j^*$ with covariates corresponding to $\tilde{m}_\beta \cup m_{\beta^*}$. 
Let $\delta>0$ be the minimum value of  the true absolute coefficient vector over all $P_j^*$. Then corresponding least square estimates for $\tilde{m}_\beta$,   for coefficients  present in true model has absolute less than $3\delta/4$ with probability less than $e^{-c_1n}$ for some universal constant  $c_1>0$ and hence, the probability that some covariate belonging to true model has coefficient estimate less than $3\delta/4$ for some $\tilde{m}_\beta \cup m_{\beta^*}$ is less than $d_n^{q_n}e^{-c_1n}=o(\delta_n)$. Similarly, for a covariate not present in the true model the corresponding coefficient has absolute value less than $\delta/4 $ with probability $1-o(\delta_n)$ over all possible covariate combination. Hence, 
\begin{eqnarray}
P(R_{m}^{(j)}-R_{m'}^{(j)}\succeq n, \text{ for all covariate combination })\geq 1-\delta_n,\nonumber \\
P(R_{m}^{(j)}-R_{t}^{(j)}\succeq n, \text{ for all covariate combination })\geq 1-2\delta_n.
\label{non_nested}
\end{eqnarray}

Hence for $m_{\beta}\supset m_{\beta^*}$, with probability at least $1-\delta_n$
\[\log \Pi(M_{l^*_t,m_\beta}|\cdot) -\log \Pi(M^*|\cdot) \leq (q-t)\log \tilde{p}_n-\frac{1}{2}((q-t)(1-2\alpha {(\log n)}^{\alpha_1})\log n +d_0(q-t)\]
uniformly over $m_\beta$, where $d_0$ does not depend on $m_\beta$, and $-\log \tilde{p}_n\sim (\log n)^{1+\alpha_1}$.

Hence, choosing $\alpha>0$ small enough, summing over $m_{\beta}\supset m_{\beta^*}$,
\[\frac{\sum_{m_{\beta}\supset m_{\beta^*};\#\{m_\beta\}\leq q_n} \Pi(M_{l^*,m_\beta}|\cdot)}{\Pi(M^*|\cdot)}\leq \sum_{q> t}e^{\log (\tilde{p}_n^{(q-t)})}{d_n\choose q-t}e^{-.5(1-2\alpha {(\log n)}^{\alpha_1})(q-t)\log n}{d_1}^{q-t}\preceq \tilde{p}_n^{\alpha'} \rightarrow 0\]
as $n\rightarrow \infty$ infinity, for some $0<\alpha'<1$, choosing sufficiently small $\alpha>0$ and   here  $d_1$ is a universal constant.

Similarly, summing over  $m_{\beta} \not\supset m_{\beta^*}$ gives, for outside of a set of probability $2\delta_n$, for some constant $d'>0$, 
\[\frac{\sum_{m_{\beta}\not \supset m_{\beta^*};\#\{m_\beta\}\leq q_n} \Pi(M_{l^*_t,m_\beta}|\cdot)}{\Pi(M^*|\cdot)}\preceq e^{-d'n} \rightarrow 0.\]

As the results are shown  outside a set of probability $O(\delta_n)$ where $\delta_n\rightarrow 0$ set, this proves our claim.  

\subsection*{\underline{Part 2:$\frac{ \sum_{m_\beta:m_\beta\neq m_{\beta^*}}\Pi(M^{\xi}_{l,m_\beta}|\cdot)}{\Pi(M^*|\cdot)}\rightarrow 0$}}

Using calculation from Theorem \ref{bfctr_vrsel}, using $\xi$ in place of $\epsilon_n$, or repeating the argument of Theorem \ref{bfctr_vrsel} proof  for Case (ii) of Theorem \ref{thm1} proof, we can bound the cross product terms, and bound the Chi-square term $\|X_j\hat{\beta}^j-X_jE[\hat{\beta}^j]\|^2$ as in  Part 1 in Theorem \ref{vr_selection2} or in Theorem \ref{bfctr_vrsel}, outside a set of probability approaching zero, and therefore, outside the small probability set, we have $\log  L({\bf Y}_n|M^\xi_{l,m_\beta})- \log L({\bf Y}_n|M^*_{l_{\tilde{t}}})\preceq -n$ uniformly over covariate choices.  Here, $M^*_{l_{\tilde{t}}}$ denote the model with true covariate combination and  change points at $\tilde{t_1},\cdots,\tilde{t}_l$ for any refinement of true partition corresponding to $t^*_1,\cdots,t_{l^*}^*$.  

Hence,  we have  $\frac{\sum_{m_\beta\neq m_{\beta^*}}\Pi(M^\xi_{l,m_\beta}|\cdot)}{\Pi(M^*_{l_{\tilde{t}}}|\cdot)}\leq  \sum^{q_n}_{q= t+1}e^{\log (\tilde{p}_n^{(q-t)})}{d_n\choose q-t}e^{-c'n}\rightarrow 0$ in probability, where $c'>0$ is a constant.  We already have shown that $\frac{\Pi(M^*_{l_{\tilde{t}}}|\cdot)}{\Pi(M^*|\cdot)}\rightarrow 0$ in probability, which proves our claim. 

\subsection*{\underline{Part 3:$\frac{ \sum_{m_\beta:m_\beta\neq m_{\beta^*}}\Pi(M^{}_{l,m_\beta}|\cdot)}{\Pi(M^*|\cdot)}\rightarrow 0$; $l<l^*$}}

Using the calculation from the proof of Theorem \ref{bfctr_vrsel} we  bound the cross product terms and Chi-square term  over all covariate choice, as in last part, and using Case (iv) in Theorem \ref{thm1} proof,  outside a set with probability approaching zero, $\log  L({\bf Y}_n|M_{l,m_\beta})- \log L({\bf Y}_n|M^*)\preceq -n$ uniformly over covariate choices. Hence, $\frac{ \sum_{m_\beta:m_\beta\neq m_{\beta^*}}\Pi(M^{}_{l,m_\beta}|\cdot)}{\Pi(M^*|\cdot)}\preceq e^{-c''n}\rightarrow 0$ in probability, where $c''>0$.

\section*{Proof of Theorem \ref{bfctr_vrsel}}
Note that for a subset $\mathscr{S}$ of $i=1,\cdots, n$, of cardinality $n_s$ and for $ P(|\frac{1}{\sqrt{n_s}}\sum_{i \in \mathscr{S}}e_ix_{ij}|\geq \sqrt{q}\log n) \leq e^{-cq(\log n)^2}$ for a universal constant $c>0$, for $q$ many covariates. Similar bound can be derived for each coordinate for $\sqrt{n}(\hat{\beta}^j-E[\hat{\beta}^j])$. Hence, over all possible covariate and change point choices $|E_1|, |E_2|,|E_3|\preceq  max\{\sqrt{n}\epsilon_nq^2_n\sqrt{q_n}\log n,\sqrt{n\epsilon_n}q_n\sqrt{q_n}\log n\}$ outside a set of probability at most $\tilde{\delta}^n\sim n^{l^*}\sum_{q=1}^{q_n} d_n^qe^{-c q(\log n)^2}\rightarrow 0$.  

 Again, $\|X_j\hat{\beta}^j-X_jE[\hat{\beta}^j]\|^2\leq q+ q(\log n)^{1+\alpha_1}$ outside a set of probability approaching zero from Part 1 of Theorem \ref{vr_selection2}. Therefore, outside a set of probability approaching zero,
\begin{align*}
-\sigma^2\log L({\bf Y}_n|M^{\epsilon_n}_{l^*,m_\beta})&+\sigma^2 \log L({\bf Y}_n|M^*)\succeq \\
& n\epsilon_n -\sqrt{n}\epsilon_n q_n^{2.5}\log n -\sqrt{n\epsilon_n}q^{1.5}_n \log n+q_n\log n\rightarrow \infty.
\end{align*}
Note that in the above equation, for $\epsilon_n$,  we have an universal lower bound of the order of $n\epsilon_n$ for the quadratic term corresponding to $ \frac{1}{2}\sum_{j:P_j\cap P_j^*\neq \phi}({\tilde{\beta}}^{*j}-E[\hat{\beta}^j])'X_{j\cap j^*}'X_{j\cap j^*}({\tilde{\beta}}^{*j}-E[\hat{\beta}^j])+ \frac{1}{2}\sum_j\sum_{k\neq j: P_j\cap P_k^*\neq \phi}({\tilde{\beta}}^{*k}-E[\hat{\beta}^j])'X'_{j\cap k^*}X_{j\cap k^*}(\tilde{\beta}^{*k}-E[\hat{\beta}^j])$  from equations \eqref{decomp11} and \eqref{rate_decomp} over all covariate and change point choices, as a result of the Eigen value condition given in $A6$.

Hence, summing over possible covariate choices
\[\sum_{m_\beta \supset m_{\beta^*} }P_n(M^{\epsilon_n}_{l^*,m_\beta}:M^*)\preceq q_ne^{-\alpha'n\epsilon_n +q_n\log d_n}\rightarrow 0\] in probability uniformly over change point choices,  where $\alpha'>0$ is a constant.

Similarly, for $m_\beta \not \supset m_{\beta^*}$, outside a set with probability approaching zero, we have  \begin{eqnarray*}
-\sigma^2\log L({\bf Y}_n|M^{\epsilon_n}_{l^*,m_\beta})&+&\sigma^2 \log L({\bf Y}_n|M^*)\succeq n, \end{eqnarray*}
from calculation similar to that of leading to   equation \eqref{non_nested}.
Hence, $\sum_{m_\beta \not \supset m_{\beta^*} }P_n(M^{\epsilon_n}_{l^*,m_\beta}:M^*)\rightarrow 0$  which concludes our claim.
\subsection*{Addressing Remark \ref{diff_var2}}
Remark \ref{diff_var2} follows from the proofs of Theorems \ref{thm1_emp} and \ref{thm2_emp}.

\section*{Proof of Theorem \ref{vr_selection1}}
Proof of Theorem \ref{vr_selection1} follows directly from the proof of Theorem \ref{vr_selection2}.

\bibliographystyle{biom}
\bibliography{changepoint_refs,hs-review}

\end{document}